\documentclass[showpacs,showkeys,preprintnumbers,pre,amsmath,amssymb,superscriptaddress,twocolumn]{revtex4-2}
\usepackage[english]{babel}
\usepackage[utf8]{inputenc}
\usepackage{amsthm} 
\usepackage{color}
\usepackage{dcolumn}
\usepackage{tikz-cd}

\usepackage[colorlinks=true,linkcolor=blue,urlcolor=blue,citecolor=blue]{hyperref}
\usepackage{bm}

\usepackage{algorithm,algorithmicx,algpseudocode}

\newtheorem{theorem}{Theorem}[section]
\newtheorem{corollary}[theorem]{Corollary}
\newtheorem{lemma}[theorem]{Lemma}

\begin{document}

\title[Honeycomb evaluations]{Exact Evaluation of Hexagonal Spin-networks and Topological Quantum Neural Networks}

\author{Matteo Lulli}
\address{Department of Mechanics and Aerospace Engineering,
	Southern University of Science and Technology, Shenzhen, Guangdong 518055, China}
\email{lulli@sustech.edu.cn}

\author{Antonino Marciano} 
\address{Department of Physics, Fudan University, 200433 Shanghai, China}
\address{Laboratori Nazionali di Frascati INFN, Frascati (Rome), Italy, EU}
\address{INFN sezione Roma Tor Vergata, I-00133 Rome, Italy, EU}
\email{marciano@fudan.edu.cn}

\author{Emanuele Zappala} 
\address{Department of Mathematics and Statistics, Idaho State University\\
Physical Science Complex |  921 S. 8th Ave., Stop 8085 | Pocatello, ID 83209, USA} 
\email{emanuelezappala@isu.edu}

\begin{abstract}
The physical scalar product between spin-networks has been shown to be a fundamental tool in the theory of topological quantum neural networks (TQNN), which are quantum neural networks previously introduced by the authors in the context of quantum machine learning. However, the effective evaluation of the scalar product remains a bottleneck for the applicability of the theory.
We introduce an algorithm for the evaluation of the physical scalar product defined by Noui and Perez between spin-network with hexagonal shape. By means of recoupling theory and the properties of the Haar integration we obtain an efficient algorithm, and provide several proofs regarding the main steps. We investigate the behavior of the TQNN evaluations on certain classes of spin-networks with the classical and quantum recoupling. All results can be independently reproduced through the ``idea.deploy" framework~\href{https://github.com/lullimat/idea.deploy}{\nolinkurl{https://github.com/lullimat/idea.deploy}} 
\end{abstract}

\maketitle

\section{Introduction}\label{sec:intro} 
The high computational demand from every sector of contemporary science, including particle physics and condensed matter, has propelled the investment in new approaches. These have arguably become the holy grail of scientific computation, e.g. quantum computing. In turn, quantum computational approaches leave the unanswered question of how to process the data in quantum machines such as quantum computers.

Important recent developments in deriving novel and efficient algorithms in quantum machine learning have been rooted in the theoretical foundation of either quantum  mechanics \cite{biamonte2017quantum,farhi2018classification,Beer_2020} or its extension to continuous systems,  quantum field theory \cite{TQNN,Halverson_2021,grosvenor2022edge,Fields:2022bti,marciano2022deep}. These attempts constitute the answer to the need of quantum algorithms for quantum computing, and the reason to propose quantum neural networks (QNN) --- see e.g. \cite{Beer_2020} --- and their extensions in the continuum \cite{TQNN,Fields:2022bti,marciano2022deep}.

A prototype for Universal Quantum Computation is provided by the Reshetikhin-Turaev model \cite{reshetikhin1991invariants}, as proved by Freedman-Kitaev-Wang \cite{key1911253m,
key1910832m}. More recently, topological quantum neural networks (TQNN), based on the TQFTs such as the Turaev-Viro model \cite{turaev1992state} and its physically motivated generalizations, have been proposed as a candidate to provide quantum algorithms in quantum computing. The advantage of TQNNs lies in the fact that they share a common ground with material science, and in particular with the string-net models of Levin-Wen \cite{2003PhRvB..67x5316L, Levin:2004mi}.

This thread of thoughts motivates us in believing that a successful translation of  the approach by Freedman-Kitaev-Wang in our TQFT methods, which is known to be possible at the mathematical level (\cite{roberts1995skein,benedetti1996roberts,lickorish1991three,lickorish1993skeins,lickorish1992calculations}) and are at the base of the TQNNs introduced in \cite{TQNN,Fields:2022bti,marciano2022deep}, will result in a Universal Quantum Computing that is implementable in practice in material science.

This is achieved through the equivalent language of string-nets \cite{2003PhRvB..67x5316L, Levin:2004mi}, providing an alternative to topological quantum computing with anyons. The tight mathematical connection relating Reshetikhin-Turaev model and Turaev-Viro model \cite{roberts1995skein,benedetti1996roberts,lickorish1991three,lickorish1993skeins,lickorish1992calculations} (one is known to be the ``square root'' of the other) allows to use our methods based on the latter to recast the former in terms of string-nets, for a material-science concrete implementation through equivalence between spin-nets and Turaev-Viro model \cite{kirillov2011string,runkel2020string,kadar2010microscopic,koenig2010quantum}, rather than their traditional anyonic-based language.

TQNNs are represented as spin-network states supported on graphs \cite{TQNN,Fields:2022bti,marciano2022deep}. These are one-complexes defined as the dual simplicial complexes to the boundaries of a manifold.  Spin-networks represent then boundary states (input/output data). The intrinsic quantumness of TQNNs stands in the fact that the dynamical evolution or these boundary states is attained through the sum over an infinite amount of intermediate virtual states (filters/hidden layers). This is the key element to the derivation of novel (quantum) algorithms. The latter are in principle characterized by higher accuracy and less computational time than traditional deep neural networks (DNN) ones, thus more adapt to machine implementations.

Within this framework, it becomes then urgent to obtain the exact evaluation of spin-networks. This is a problem that requires, in principle, an exponential time complexity. In fact, the recoupling theory defined by Kauffman and Lins \cite{KL} defines a partition function from spin-networks by summing over all possible combinations of admissible colorings, and is based on the (factorial) unraveling of the Jones-Wenzl projector \cite{KL,Lick}. 

Recoupling theory was originally introduced to define topological invariants of $3$-manifolds. In fact, one could show that the aforementioned partition function defined on spin-networks dual to the cells of a (regular enough) simplicial decomposition of a $3$-manifold is invariant under Matveev-Piergallini moves \cite{piergallini1986standard,matveev1988transformations}, ensuring that the numerical value of the partition function is unchanged when considering homeomorphic topological spaces. 

The theory has become widely applied in quantum gravity, where it has played a central role in the formulation by Perez and Noui \cite{NP} of the physical inner product for Euclidean quantum gravity in 3-dimensions, achieved via the regularization of the projector that imposes the curvature constraint of $SU(2)$ symmetric $BF$ theory at the quantum level. More recently, the implementation of a projector similar to the one studied by Perez and Noui, applied to a still topological extended $BF$ theory provided with cosmological constant, has been derived in \cite{CosmK}. There, it has been shown that the imposition of the curvature constraint with cosmological constant naturally provides the recoupling theory of a quantum group to emerge from the initial $SU(2)$ symmetry structure. This has finally allowed to introduce recoupling theory of quantum groups in 3-dimensional quantum gravity in a constructive way, explaining the emergence of the recoupling theory of $SU_q(2)$ from that one of $SU(2)$.    

The recoupling theories of $SU(2)$ and $SU_q(2)$ are crucial for the applications into quantum machine learning that were explored in \cite{TQNN,Fields:2022bti,marciano2022deep}. As we anticipated, the notion of TQNNs is formulated by means of a TQFT, and is in practice evaluated via recoupling. Although in \cite{TQNN,Fields:2022bti,marciano2022deep} concrete examples were provided only accounting for the recoupling theory of $SU(2)$, a natural extension to quantum groups, and in particular to the recoupling theory of $SU_q(2)$, can be envisaged following the constructive arguments deployed in \cite{CosmK}.  

Nonetheless, the main bottleneck of the concrete applicability of the results in \cite{TQNN,Fields:2022bti,marciano2022deep} remains the ability of evaluating the Perez-Noui projector in an efficient manner. As a subcase, this also includes the problem of evaluating spin-networks in general form, which is a notoriously complicated problem and it has previously been considered in the seminal articles \cite{SN_evaluation,BL}, where theoretical and computational results regarding certain specific cases have been considered in detail. We focus in this article on the evaluation of spin-networks of hexagonal shape and arbitrary size, and relate these objects to the pixel space of images to apply TQNNs. We use these results to obtain an algorithm for the evaluation of the Perez-Noui projector on $SU(2)$ \cite{NP}, and its generalization to $SU_q(2)$ \cite{CosmK}.

The plan of the paper is the following. In Sec.~\ref{pitoex} we delve into the correspondence between the pixel space of images and the hexagonal spin-networks. In Sec.~\ref{sec:Honey} consider spin-networks that are obtained by juxtaposition of hexagonal cells. In Sec.~\ref{algo} we provide the algorithm for the evaluation of the spin-network. In Sec.~\ref{sec:comput} we compute the transition amplitudes between two different hexagonal spin-networks. In Sec.~\ref{phase-space} we show some numerical results for the transition probability between two different hexagonal spin-networks. In Sec.~\ref{Isi} we comment on the relation with the Ising model. Finally, in Sec.~\ref{conclu} we provide outlooks for future investigations and preliminary  conclusions.

\section{From pixel space to hexagonal spin-networks} \label{pitoex}
\begin{figure*}[ht!]
\includegraphics[width=.55\linewidth]{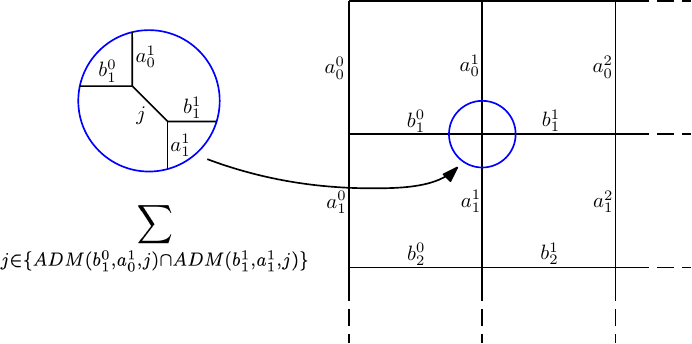}
\caption{Example of spin-network with grid support where desingularization is performed at each $4$-valent vertex. Here, the spin colors $j$ shown in the zoomed in part of the figure, run over all the compatible colors with respect to the incoming edges.}
\label{fig:desing_grid}
\end{figure*}

Our starting point is a correspondence between the pixel space of images, and hexagonal spin-networks. This also motivates our interest in evaluating hexagonal spin-networks, as they are seen to correspond to images, therefore constituting our key to translate data sets into the input of TQNNs.  

We start our discussion by considering first a very natural approach that rapidly incurs into an unwanted computational overhead. We consider an $n\times n$ grid where each square indicates a pixel. Each pixel is endowed with a label between $0$ and $m$ indicating the intensity of the black color. It is clear that in this way we can represent a black and white image of $n\times n$ resolution. To such an image, we can associate a spin-network proceeding as follows. Let $P_k$ denote the $k^{\rm th}$ pixel of the grid in the lexicographical order. We introduce the barycenter coordinate of each pixel (square in the grid), and consider the von Neumann neighborhood $\mathcal N_k$ of $P_k$, which is given by $\mathcal N_k = \{P_{k-1}, P_{k+1}, P_{k-n}, P_{k+n}\}$ with the assumption that one or two of the pixels in $\mathcal N_k$ is omitted for pixels $P_k$ along the edges or the corners, respectively. We observe that we do not use periodic boundaries here, so that our resulting spin-networks do not lie in the torus, but in the plane. The centers of $P_k$, which we denote by $C_k$, will be the vertices of the spin-networks, and each $C_k$ is connected to all the vertices corresponding to pixels belonging to its von Neumann neighborhood. The colors of the spin-networks are attributed by labeling the edges between the vertices based on the difference of the pixel values at the vertices $C_k$ and $C_l$ that they connect. This approach was followed for instance in \cite{TQNN}. 

However, while working in the semi-classical limit does not incur in any problems (see e.g. \cite{TQNN}), when we try to evaluate the spin-networks obtained through this procedure we find that each vertex needs to be desingularized as shown in Figure~\ref{fig:desing_grid}, in order to obtain two trivalent vertices from each $4$-valent vertex. Each desingularization will introduce a summation over the admissible colors, and this negatively affects the computational cost of a TQNN algorithm based on spin-networks with such grid supports. 





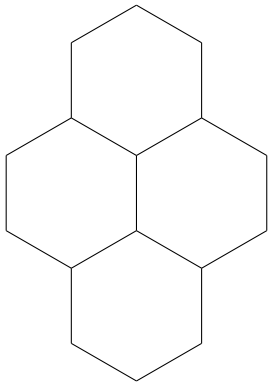
\begin{figure}[h!]
\begin{center}

\begin{tikzpicture}[rotate=30]
\foreach \a in {0,120,-120} \draw (3*1,2*sin{60}*1) -- +(\a:1);
\foreach \a in {0,120,-120} \draw (3*1,2*sin{60}*2) -- +(\a:1);
\foreach \a in {0,120,-120} \draw (3*0+3*cos{60},2*sin{60}*1+sin{60}) -- +(\a:1);
\foreach \a in {-120,120} \draw (3*1+3*cos{60},2*sin{60}*1+sin{60}) -- +(\a:1);
\foreach \a in {0,-120} \draw (3*0+3*cos{60},2*sin{60}*2+sin{60}) -- +(\a:1);
\foreach \a in {-120} \draw (3*1+3*cos{60},2*sin{60}*2+sin{60}) -- +(\a:1); 
\foreach \a in {120} \draw (3*1+3*cos{60},2*sin{60}*2+sin{60}) -- +(\a:1);  
\foreach \a in {0} \draw (3*1,2*sin{60}*3) -- +(\a:1);
\foreach \a in {-120} \draw (3*1,2*sin{60}*3) -- +(\a:1);
\foreach \a in {0} \draw (3*0+3*cos{60},2*sin{60}*0+sin{60}) -- +(\a:1);    
\foreach \a in {120} \draw (3*0+3*cos{60},2*sin{60}*0+sin{60}) -- +(\a:1);  
\end{tikzpicture}
\end{center}

\caption{Honeycomb lattice with size two side, where corners and vertices are intertwiners.}
\label{fig:honey}
\end{figure}

Instead, we proceed by considering a honeycomb lattice structure as in Figure~\ref{fig:honey}. It is clear that one can find a one-to-one correspondence between hexagons in the lattice in the figure and a $2\times 2$ (pixel) image. For the $n\times n$ pixel space one proceeds analogously. This process allows us to associate to a figure with $n\times n$ pixel resolution a hexagonal lattice which we will call $n\times n$ as well.  
Using a scheme similar to the one described above, we can associate to each pixel in black and white or RGB colors a numerical value between $0$ and some upper bound $N$ depending on the coloring scale. Each perimeter of the hexagon is then given the ``color'' $r\in [0,N]$ determined by the pixel color. On edges that are shared among hexagons, the colors will be summed. So, if the edge $e$ is shared between hexagon $h_i$ and $h_j$ with respective colors $r_i$ and $r_j$, we have that $e$ takes the color $r_i+r_j$. At each edge we now associate two projectors (which is the same one as by definition of projector) with the implicit assumption that each edge is labeled by a number of strands that derived by summing pixel colors. Using the definition of spin-network as in \cite{KL}, we can rewrite the whole hexagon lattice as a spin-network as in Figure~\ref{fig:pixel_hex}, where the $2\times 2$ case is depicted.

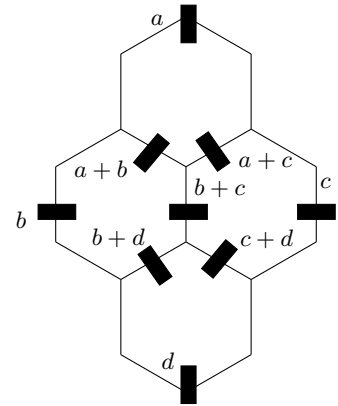
\begin{figure}[h!]
\begin{tikzpicture}[rotate=30]
\foreach \a in {0,120,-120} \draw (3*1,2*sin{60}*1) -- +(\a:1);
\foreach \a in {0,120,-120} \draw (3*1,2*sin{60}*2) -- +(\a:1);
\foreach \a in {0,120,-120} \draw (3*0+3*cos{60},2*sin{60}*1+sin{60}) -- +(\a:1);
\foreach \a in {-120,120} \draw (3*1+3*cos{60},2*sin{60}*1+sin{60}) -- +(\a:1);
\foreach \a in {0,-120} \draw (3*0+3*cos{60},2*sin{60}*2+sin{60}) -- +(\a:1);
\foreach \a in {-120} \draw (3*1+3*cos{60},2*sin{60}*2+sin{60}) -- +(\a:1); 
\foreach \a in {120} \draw (3*1+3*cos{60},2*sin{60}*2+sin{60}) -- +(\a:1);  
\foreach \a in {0} \draw (3*1,2*sin{60}*3) -- +(\a:1);
\foreach \a in {-120} \draw (3*1,2*sin{60}*3) -- +(\a:1);
\foreach \a in {0} \draw (3*0+3*cos{60},2*sin{60}*0+sin{60}) -- +(\a:1);    
\foreach \a in {120} \draw (3*0+3*cos{60},2*sin{60}*0+sin{60}) -- +(\a:1);  


\draw[fill=black,rotate=185] (-2.1,-2.1) rectangle (-2.3,-2.6);
\draw[fill=black,rotate=60] (3.8,-0.65) rectangle (4,-1.15);
\draw[fill=black,rotate=110] (1,-3.1) rectangle (1.2,-3.6);
\draw[fill=black,rotate=110] (2.6,-3.6) rectangle (2.8,-4.1);
\draw[fill=black,rotate=185] (-3.6,-2.9) rectangle (-3.8,-3.4);
\draw[fill=black,rotate=60] (3.8,1.1) rectangle (4,0.6);
\draw[fill=black,rotate=60] (3.8,-2.35) rectangle (4,-2.85);
\draw[fill=black,rotate=150] (-1,-1.35) rectangle (-0.8,-1.85);
\draw[fill=black,rotate=150] (-1,-6.15) rectangle (-0.8,-6.65);

\node (a) at (3.65,5.35) {$a$}; 
\node (a) at (2.,4.) {$a+b$}; 
\node (a) at (3.95,3.) {$a+c$}; 
\node (a) at (3.25,3.) {$b+c$};
\node (a) at (4.5,2.35) {$c$};
\node (a) at (.75,3.95) {$b$};
\node (a) at (1.75,3.1) {$b+d$};
\node (a) at (3.45,2.1) {$c+d$};
\node (a) at (1.5,1.35) {$d$};
\end{tikzpicture}
\caption{Honeycomb lattice with size two side corresponding to a $2\times 2$ pixel figure: the values $a, b, c, d$ at the center of the hexagons represent the colors corresponding to the pixel color, and the projectors are labeled by the number of strands obtained by summing the pixel colors}
\label{fig:pixel_hex}
\end{figure}

\section{Honeycomb spin-networks and their evaluation}\label{sec:Honey}
    
We consider spin-networks that are obtained by juxtaposition of hexagonal cells, where each vertex is trivalent, as depicted in Figure~\ref{fig:honey}, where a four cell honeycomb is shown. In other words, we consider a honeycomb lattice whose vertices are intertwiners, and whose edges are bundles (i.e. tensor products) of $\frak{su}_2(\mathbb C)$ fundamental representations symmetrized by the {\it Jones-Wenzl idempotent}, which we will also call {\it symmetrizer}. We denote by the symbol $\mathcal H_n(\bar a,\bar b,\bar c, \bar d, \bar e)$ the square honeycomb lattice whose side is of size $n$ and whose edges are labelled by spin-colors $\bar a,\bar b,\bar c, \bar d, \bar e$, following a precise scheme that will be described later in the article. Here $\bar a$ etc., indicate vectors of spin colors associated to the edges of the spin-networks. When the spin colors do not play a role in the discussion, or if there is no risk of confusion, we will omit to write the labels and will content ourselves with simply writing $\mathcal H_n$. In Figure~\ref{fig:honey}, for example, a square honeycomb lattice of side $n = 2$ is represented. The labels are not assumed to constitute admissible triples a priori, and we set to zero the evaluation of a honeycomb spin-network whose labels contain a non-admissible triple at some vertex. In this article we allow, albeit rather improperly, spin-networks with open ends, i.e. supported on graphs that have edges with one endpoint not connected to any vertex. Considering these types of spin-networks simplifies certain inductive procedures in the constructions, as we shall see in the next results. They will be referred to as {\it open-end} or {\it open-edge} spin-networks, in the rest of this article. 

Along with the spin-networks $\mathcal H_n$, we also define the open-end spin-networks $\mathcal O_n$ as follows. For each $n$, $\mathcal O_n$ is defined as a single hexagonal cell, where we attach three open spin-network edges, symmetric with respect to the hexagonal cell. The central edge is a single edge, while the two lateral edges are assumed to consist of $2n-1$ connected edges according to the geometry depicted in Figure~\ref{fig:lateral}, where there are $n-1$ vertical edges and $n$ horizontal ones.

\begin{figure}[t!]
\begin{center}
\begin{tikzpicture}
\draw (0,0) -- (2,0);
\draw[dashed] (2,0) -- (4,0);
\draw (4,0) -- (6,0);
\draw (1,0) -- (1,-1);
\draw (5,0) -- (5,-1);
\draw[fill=black] (1,0) circle (2pt);
\draw[fill=black] (5,0) circle (2pt);
\end{tikzpicture}
\caption{Lateral open-end spin-networks of $\mathcal O_n$}
\label{fig:lateral}
\end{center}
\end{figure}
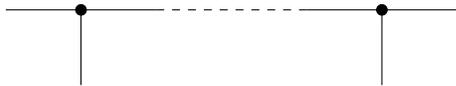

The open-end spin-network $\mathcal O_n$ is depicted in Figure~\ref{fig:octopus}.
\begin{figure*}[ht!]
\begin{center}
    \begin{tikzpicture}[rotate=30]
    \foreach \a in {0,120,-120} \draw (3*1,2*sin{60}*1) -- +(\a:1);
    \foreach \a in {0,-120} \draw (3*1,2*sin{60}*2) -- +(\a:1);
\foreach \a in {120,-120} \draw (3*1+3*cos{60},2*sin{60}*1+sin{60}) -- +(\a:1);

\draw[rotate=-30] (2.6,3.5) -- (4.6,3.5);
\draw[dashed,rotate=-30] (4.6,3.5) -- (6.6,3.5);
\draw[rotate=-30]  (6.6,3.5) -- (8.6,3.5);
\draw[rotate=-30]  (3.6,3.5) -- (3.6,2.5);
\draw[rotate=-30]  (7.6,3.5) -- (7.6,2.5);
\draw[fill=black,rotate=-30] (3.6,3.5) circle (2pt);
\draw[fill=black,rotate=-30] (7.6,3.5) circle (2pt);

\draw[rotate=-30] (2.6-7.7,3.5) -- (4.6-7.7,3.5);
\draw[dashed,rotate=-30] (4.6-7.7,3.5) -- (6.6-7.7,3.5);
\draw[rotate=-30]  (6.6-7.7,3.5) -- (8.6-7.7,3.5);
\draw[rotate=-30]  (3.6-7.7,3.5) -- (3.6-7.7,2.5);
\draw[rotate=-30]  (7.6-7.7,3.5) -- (7.6-7.7,2.5);
\draw[fill=black,rotate=-30] (3.6-7.7,3.5) circle (2pt);
\draw[fill=black,rotate=-30] (7.6-7.7,3.5) circle (2pt);

\draw[fill=black][rotate=-30] (2.6,3.5) circle (2pt);
\draw[fill=black][rotate=-30] (0.9,3.5) circle (2pt);
    \end{tikzpicture}
    \caption{Open-edge spin-network $\mathcal O_n$}
    \label{fig:octopus}
\end{center}
\end{figure*}

Let $\mathcal N$ denote a spin-network, and let $\mathcal L$ denote an open-end spin-network, with legs labeled $a_1, \cdots, a_r$, for some $r\in \mathbb N$. Let $\bar v = (v_1, \ldots, v_r)$ denote a list of vertices of $\mathcal N$. Then, we can define the composition, written $\mathcal N \circ_{\bar v} \mathcal L$, where each edge $a_i$ of $\mathcal L$ is joined with the vertex $v_i$ of $\mathcal N$. If the edges are colored by spin colors, then we set to zero the composition of networks where the colors are not admissible, while we denote the admissible composition by the same symbol as above. 
Then we have the following result. It holds that 
\begin{eqnarray}
     \mathcal H_{n+1} = \mathcal H_n \circ_{\bar v} \mathcal O_n,
\end{eqnarray}
for every  $n\in \mathbb N$, and for some choice of vertices $\bar v$ in $\mathcal H_n$ (see Lemma~\ref{lem:composition}). 

\begin{figure}[h!]
\begin{tikzpicture}
\draw (0,0) -- (1,0);
\draw (1,0) ..controls (2,0.75).. (3,0);
\draw (1,0) ..controls (2,-0.75).. (3,0);
\draw (3,0) -- (4,0);
\end{tikzpicture}
\caption{Bubble graph}
\label{fig:bubble}
\end{figure}

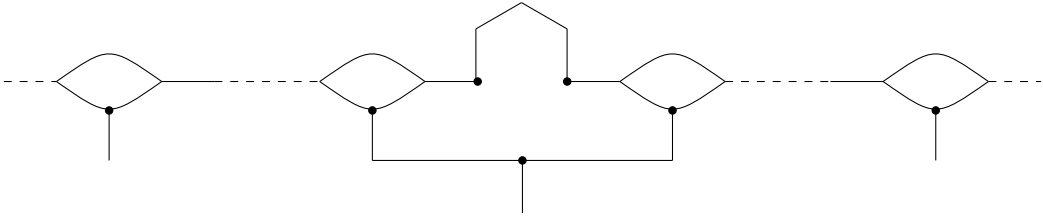
\begin{figure*}[ht!]
\begin{tikzpicture}[rotate=30,scale=0.7]
\foreach \a in {0,-120} \draw (3*1,2*sin{60}*2) -- +(\a:1);
\foreach \a in {120,-120} \draw (3*1+3*cos{60},2*sin{60}*1+sin{60}) -- +(\a:1);

\draw[rotate=-30] (2.6,3.5) -- (3.6,3.5);
\draw[rotate=-30] (3.6,3.5) .. controls (4.6,4.2) .. (5.6,3.5);
\draw[rotate=-30] (3.6,3.5) .. controls (4.6,2.8) .. (5.6,3.5);
\draw[dashed,rotate=-30] (5.6,3.5) -- (7.6,3.5);
\draw[rotate=-30] (7.6,3.5) -- (8.6,3.5);
\draw[rotate=-30] (8.6,3.5) .. controls (9.6,4.2) .. (10.6,3.5);
\draw[rotate=-30] (8.6,3.5) .. controls (9.6,2.8) .. (10.6,3.5);
\draw[dashed,rotate=-30] (10.6,3.5) -- (11.6,3.5);
\draw[rotate=-30] (4.6,3) -- (4.6,2);
\draw[rotate=-30] (9.6,3) -- (9.6,2);
\draw[fill=black][rotate=-30] (4.6,2.95) circle (2pt);
\draw[fill=black][rotate=-30] (9.6,2.95) circle (2pt);

\draw[dashed,rotate=-30] (2.6-10.7,3.5) -- (3.6-10.7,3.5);
\draw[rotate=-30] (3.6-10.7,3.5) .. controls (4.6-10.7,4.2) .. (5.6-10.7,3.5);
\draw[rotate=-30] (3.6-10.7,3.5) .. controls (4.6-10.7,2.8) .. (5.6-10.7,3.5);
\draw[rotate=-30] (5.6-10.7,3.5) -- (6.6-10.7,3.5);
\draw[dashed,rotate=-30] (6.6-10.7,3.5) -- (8.6-10.7,3.5);
\draw[rotate=-30] (8.6-10.7,3.5) .. controls (9.6-10.7,4.2) .. (10.6-10.7,3.5);
\draw[rotate=-30] (8.6-10.7,3.5) .. controls (9.6-10.7,2.8) .. (10.6-10.7,3.5);
\draw[rotate=-30] (10.6-10.7,3.5) -- (11.6-10.7,3.5);
\draw[rotate=-30] (4.6-10.7,3) -- (4.6-10.7,2);
\draw[rotate=-30] (9.6-10.7,3) -- (9.6-10.7,2);
\draw[fill=black][rotate=-30] (4.6-10.7,2.95) circle (2pt);
\draw[fill=black][rotate=-30] (9.6-10.7,2.95) circle (2pt);

\draw[fill=black][rotate=-30] (11.6-10.7,3.5) circle (2pt);
\draw[fill=black][rotate=-30] (2.6,3.5) circle (2pt);

\draw[rotate=-30] (9.6-10.7,2) .. controls (2,2) .. (4.6,2); 
\draw[rotate=-30,fill=black] (9.6-10.7+2.85,2) circle (2pt);
\draw[rotate=-30] (9.6-10.7+2.85,2) -- ((9.6-10.7+2.85,1);
\end{tikzpicture}
\caption{Spin-network $\mathcal{BO}_n$}
\label{fig:bubble_octopus}
\end{figure*}

For a spin-network composition as above (and in the statement of Lemma~\ref{lem:composition}), we say that the spin-network components $\mathcal H_n$ and $\mathcal O_n$ {\it inherit} the labels from the larger spin-network $\mathcal H_{n+1}$, if the spin colors of the components coincide with the respective ones in $\mathcal H_{n+1}$. When the vertices that are used for the composition are clearly understood, and there is no need to remark how the composition is being performed, we simply write the symbol $\circ$ without indicated the vector $\bar v$ of vertex indices.

We now define the following type of spin-networks, denoted by $\mathcal{BO}_n$, and obtained from the graph supporting $\mathcal O_n$ by replacing each lateral vertex by a {\it bubble graph} depicted in Figure~\ref{fig:bubble}, as well as deleting the lower half of the hexagonal edge, and connecting the first two lateral vertical edges. The graph $\mathcal{BO}_n$ is represented in Figure~\ref{fig:bubble_octopus}. Lastly, let $\mathcal {HH}_n$ denote the spin-network obtained from $\mathcal H_n$ by deleting the hexagons along the upper perimeter. For $\mathcal H_2$, for example, this means that one deletes the top hexagon, while for $\mathcal H_3$ one deletes the top $3$ hexagons and so on. For $n=1$ we set $\mathcal{HH}_1$ to consist of a single edge corresponding to the lower perimeter of the hexagon $\mathcal H_1$.

We now set a useful convention on the spin colors labeling edges of the spin-networks $\mathcal H_n$, proceeding inductively on $n$.
We start by setting the labels of the hexagon $\mathcal H_1$ as in Figure~\ref{fig:H1_labeling}.
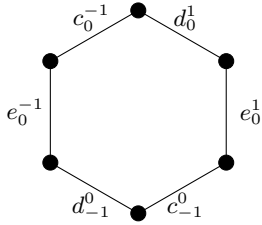
\begin{figure}
\begin{tikzpicture}[rotate=90,scale=0.5]
   \newdimen\R
   \R=2.7cm
   \draw (0:\R) \foreach \x in {60,120,...,360} {  -- (\x:\R) };
   \foreach \x in {60,120,...,360} \node[inner sep=2pt,circle,draw,fill] at (\x:\R) {};
   \node (a)  at (-2.5,-1.25) {$c_{-1}^0$};
   \node (a)  at (-2.5,1.25) {$d_{-1}^0$};
   \node (a)  at (0,3) {$e_0^{-1}$};
   \node (a)  at (0,-3) {$e_0^1$};
   \node (a)  at (2.5,-1.25) {$d_0^1$};
   \node (a)  at (2.5,1.25) {$c_0^{-1}$};
\end{tikzpicture}
    \caption{Labeling of the edges of $\mathcal H_1$}
    \label{fig:H1_labeling}
\end{figure}

Then, in the decomposition $\mathcal H_n = \mathcal H_{n-1} \circ_{\bar v} \mathcal O_{n-1}$, where $n \geq 2$, we number the edges of $\mathcal H_n$ identified with the vertical open edges of $\mathcal O_n$ as follows. The central edge is numbered $0$, then the left branch of $\mathcal O_n$ is numbered in increasing order from center to left with odd numbers, while the right branch is numbered in the same way, but with even numbers. At each configuration as in Figure~\ref{fig:k_configuration}, we indicate the five spin colors involved as $a^{\bullet}_k, b^{\bullet}_k, c^{\bullet}_k, d^{\bullet}_k, e^{\bullet}_k$, and denote the corresponding spin-network by $S^{\bullet}_k$, where $\bullet$ is a placeholder for an arbitrary index. Here, the subscript indicates the level in which the spin-network portion appears. Level $k$, indicates that it is part of the $k+1$ spin-network $\mathcal H_{k+1}$, but it does not lie in the copy of $\mathcal H_k$ inside $\mathcal H_{k+1}$ according to Lemma~\ref{lem:composition}. We will also use another index, which will appear as a superscript, to indicate the position of the spin-network portion within a level. The convention is the following. For levels where an odd number of $e_k$'s appear, we denote the central $e_k$ as $e^0_k$, while those $e_k$'s that lie on the left will be labeled $e_k^{-i}$, and those on the right $e_k^i$, in a symmetric fashion, and with increasing value of $i$ as the $e_k$'s are farther from the center. For levels with even number of $e_k$'s, we omit the central $e_k^0$ and follow the same scheme. 
Observe that for each $k$ we have that some of the edges of spin-networks $S^{\bullet}_k$ of different levels are connected, and therefore the corresponding labels are identified. In this case, we follow the convention that if $S^{\bullet}_k$ and $S^{\bullet}_{k-1}$ meet, the connecting edge will take the label of $S^{\bullet}_{k-1}$, while if $S^{\bullet}_k$ meets another $S^{\bullet}_k$, then the labels reported are those with lower order with respect to the natural lexicographical order $a<b<c<d<e$.
We observe that following the previous conventions, the labels $a$ and $b$ will not appear in the spin-network $\mathcal H_n$ except in the bottom arc, where they are labeled with subscript $-1$. Along the edges of $\mathcal H_n$, there appear arcs connecting at binary vertices. These edges merge, according to the rules of spin-networks at binary vertices. The labels that we report in these cases are dictated by the following ordering. For positive superscripts (i.e. on the right side of the perimeter), we have the order $d<c<e$, while for negative superscripts we have $c<d<e$. Then, on the meeting edges, we relabel the merged edges according to the smallest element. On central cells on top and bottom of the spin-networks, we follow the convention that the largest spin-color label is preserved. The orderings in these cases are the natural ones.
Note that the only spin-colors that appear on the (lateral) perimeter are given by the letters $c,d,e$, while the central perimeter cells are just two (bottom and top), so that the rules given above exhaust all the cases.

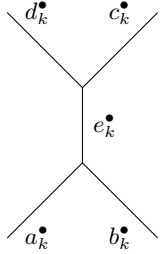
\begin{figure}
\begin{tikzpicture}
\draw (-1,1) -- (0,0) -- (1,1);
\draw (-1,-2) -- (0,-1) -- (1,-2);
\draw (0,0) -- (0,-1);
\node (a) at (-0.6,1) {$d_k^\bullet$};
\node (a) at (-0.6,-2) {$a_k^\bullet$};
\node (a) at (0.5,1) {$c_k^\bullet$};
\node (a) at (0.5,-2) {$b_k^\bullet$};
\node (a) at (0.3,-0.5) {$e_k^\bullet$};
\end{tikzpicture}   
\caption{General configuration of $S^{\bullet}_k$ at $k^{\rm th}$ edge, where $\bullet$ is a placeholder that indicates an arbitrary index}
\label{fig:k_configuration}
\end{figure}

\begin{widetext}
Now, let us define the following quantities. For spin colors $\bar a,\bar b,\bar c, \bar d, \bar e$ following the 
convention above, we define 
\begin{eqnarray*}
\Psi(\bar a, \bar b, \bar c, \bar d, \bar e\ |\ \bar i) 
&=& 
\begin{Bmatrix}
d_{\lfloor\frac{n}{2}\rfloor}^{-\lfloor\frac{n}{2}\rfloor+1} & e_{\lfloor\frac{n}{2}\rfloor+1}^{-\lfloor\frac{n}{2}\rfloor}  & i_{\lfloor\frac{n}{2}\rfloor}^{-\lfloor\frac{n}{2}\rfloor} \\
c_{\lfloor\frac{n}{2}\rfloor+1}^{-\lfloor\frac{n}{2}\rfloor-1}& d_{\lfloor\frac{n}{2}\rfloor-1}^{-\lfloor\frac{n+1}{2}\rfloor+1} & c_{\lfloor\frac{n}{2}\rfloor}^{-\lfloor\frac{n}{2}\rfloor}
\end{Bmatrix} \\
&&
\times \prod_{\lfloor \frac{n+2}{2}\rfloor -1< k \leq 2\lfloor \frac{n}{2} \rfloor} 
\begin{Bmatrix}
c_{k}^{-\lfloor \frac{n+2}{2}\rfloor + \lfloor \frac{k+1}{2} \rfloor} & e_{k}^{-\lfloor \frac{n+2}{2}\rfloor + \lfloor \frac{k+1}{2} \rfloor} & i_{k}^{-\lfloor \frac{n+2}{2}\rfloor + \lfloor \frac{k+1}{2} \rfloor} \\
c_{k}^{-\lfloor \frac{n+2}{2}\rfloor + \lfloor \frac{k+1}{2} \rfloor -1} & c_{k+1}^{-\lfloor \frac{n+2}{2}\rfloor + \lfloor \frac{k}{2} \rfloor} & d_{k}^{-\lfloor \frac{n+2}{2}\rfloor + \lfloor \frac{k+1}{2} \rfloor}
\end{Bmatrix}.
\end{eqnarray*}
Moreover, we define the ``formal involution'' $\iota$ which is applied to a symbol as given above, and acts as follows: $\iota$ exchanges the colors $d$ to $c$, it reverses the signs of the superscripts, and it leaves the subscripts the same. Applying the recoupling (\cite{KL}) we obtain that the following equality holds for all choices of compatible spin colors $a, b, c, d, e, f$:
\begin{equation*}
\begin{tikzpicture}[baseline={([yshift=0.2cm]current bounding box.center)},vertex/.style={anchor=base,
circle,fill=black!25,minimum size=18pt,inner sep=2pt}]
\draw (-1,0) -- (0,0);
\draw (0,0) .. controls (1,0.5).. (2,0);
\draw (0,0) .. controls (1,-0.5).. (2,0);
\draw (2,0) -- (3,0);
\draw (1,-0.4) -- (1,-1.4);
\draw[fill=black] (1,-0.375) circle (2pt);
\draw[fill=black] (0,0) circle (2pt);
\draw[fill=black] (2,0) circle (2pt);

\node (a) at (-0.5,0.35) {$b$};
\node (a) at (2.5,0.35) {$c$};
\node (a) at (1.7,-0.5) {$d$};
\node (a) at (0.3,-0.5) {$a$};
\node (a) at (1,0.6) {$e$};
\node (a) at (1.3,-1) {$f$};
\end{tikzpicture}
=
\begin{Bmatrix}
a & b & f\\
c & d & e 
\end{Bmatrix}
\Delta_f^{-1}
\theta(a,d,f)
\ \
\begin{tikzpicture}[baseline={([yshift=0.1cm]current bounding box.center)},vertex/.style={anchor=base,
circle,fill=black!25,minimum size=18pt,inner sep=2pt}]
\draw (-2,0) -- (0,0) -- (2,0);
\draw (0,0) -- (0,-1);

\draw[fill=black] (0,0) circle (2pt);
\node (a) at (-1,0.35) {$b$};
\node (a) at (1,0.35) {$c$};
\node (a) at (0.2,-0.7) {$f$};
\end{tikzpicture}\ \ .
\end{equation*}
This move will be referred to as ``bubble move'', for simplicity, and its proof is given in Lemma~\ref{lem:bubble-move} below.
Now, we want to show how to decompose the $\mathcal H_{n+1}$ spin-network in terms of lower degrees spin-networks of type $\mathcal {HH}_n$ and $\mathcal {BO}_n$. 
For this purpose, we decompose  $\mathcal H_{n+1}$ in a linear combination of $\mathcal {HH}_n$ and $\mathcal {BO}_n$ as
    \begin{equation}\label{eqn:half-bubble}
    \begin{aligned}
    \mathcal H_{n+1}&(\bar a,\bar b, \bar c, \bar d, \bar e) 
    =
    \Delta_{e_{n+1}^0}\theta(d_{n+1}^{0}, c_{n+1}^{0},e_{n+1}^0)
    \frac{\theta(i_{n+1}^0,c_{n+1}^{-1},c_n^{-1})}{\Delta_{i_{n+1}^0}}\frac{\theta(i_{n+1}^0,d_{n+1}^{1},d_n^{1})}{\Delta_{i_{n+1}^0}}\\
    &\times\frac{\theta(i_{\lfloor \frac{n}{2}\rfloor}^{-\lfloor \frac{n}{2}\rfloor},c_{\lfloor \frac{n}{2}\rfloor+1}^{-\lfloor \frac{n}{2}\rfloor},e_{\lfloor \frac{n}{2}\rfloor+1}^{-\lfloor \frac{n}{2}\rfloor+1})}{\Delta_{i_{\lfloor \frac{n}{2}\rfloor}^{-\lfloor \frac{n}{2}\rfloor}}}
    \frac{\theta(i_{\lfloor \frac{n}{2}\rfloor}^{\lfloor \frac{n}{2}\rfloor},d_{\lfloor \frac{n}{2}\rfloor+1}^{\lfloor \frac{n}{2}\rfloor},e_{\lfloor \frac{n}{2}\rfloor+1}^{\lfloor \frac{n}{2}\rfloor-1})}{\Delta_{i_{\lfloor \frac{n}{2}\rfloor}^{\lfloor \frac{n}{2}\rfloor}}}
    \begin{Bmatrix}
      d_{n+1}^{0}  & c_{n+1}^{-1}  & e_{n+1}^0 \\
      d_{n+1}^1  & c_{n+1}^{0}  & e_{n+2}^0
    \end{Bmatrix}\\
    &\times \mathcal {HH}_n(\bar a,\bar b,\bar c, \bar d, \bar e) \circ_{\bar v} \mathcal {BO}_n(\bar a,\bar b,\bar c, \bar d, \bar e)   
    \sum_{\bar i}\Psi(\bar a,\bar b,\bar c, \bar d, \bar e\ |\ \bar i)\iota(\Psi(\bar a,\bar b,\bar c, \bar d, \bar e\ |\ \bar i)),
    \end{aligned}
    \end{equation}
where $\Psi$ and $\iota$ have been defined above. This result is stated and proved in Lemma~\ref{lem:half-bubble}.
\end{widetext}
The coefficients $\Psi(\bar a, \bar b, \bar c, \bar d, \bar e\ |\ \bar i)$ will also be written as $\Psi_{\bar i}$ for simplicity, when it is clear what spin colors are being considered. Let $\mathcal {BO}_n$ denote an $n$-bubble spin-network as in Figure~\ref{fig:bubble_octopus}. Here we assume that the spin colors of $\mathcal {BO}_n$ are those inherited by Equation~\ref{eqn:half-bubble}. We can now apply Lemma~\ref{lem:bubble-move} on each of the bubbles of $\mathcal {BO}_n$. This will gives us $\mathcal {BO}_n$ as a sum on admissible colors of the spin-networks $\mathcal {O}_n$. The evaluation of $\mathcal {BO}_n$ is obtained through the formula
\begin{eqnarray}
    \begin{aligned}
    \lefteqn{\mathcal {BO}_n(\bar a, \bar b, \bar c, \bar d, \bar e, \bar f)}\\ &=&     
        \prod_{n-1 \leq k \leq 2n - 5}
        \begin{Bmatrix}
        c_k^{\lfloor \frac{k+1}{2}\rfloor-n - 1} & p_k^{\lfloor \frac{k+1}{2}\rfloor-n - 1} & d_k^{\lfloor \frac{k+1}{2}\rfloor-n - 2}\\
        p_{k+1}^{\lfloor \frac{k+2}{2}\rfloor-n -1} & e_{k+1}^{\lfloor \frac{k+2}{2}\rfloor-n -1} & c_{k+1}^{\lfloor \frac{k+2}{2}\rfloor-n} 
        \end{Bmatrix}\\
        &&\times \frac{\theta(c_k^{\lfloor \frac{k+1}{2}\rfloor-n - 1},e_{k+1}^{\lfloor \frac{k+2}{2}\rfloor-n -1},d_k^{\lfloor \frac{k+1}{2}\rfloor-n - 2}}{\Delta_{d_k^{\lfloor \frac{k+1}{2}\rfloor-n - 2}}}\\
        &&\times 
        \begin{Bmatrix}
        d_k^{-\lfloor \frac{k+1}{2}\rfloor + n + 1} & p_k^{-\lfloor \frac{k+1}{2}\rfloor + n + 1} & c_k^{-\lfloor \frac{k+1}{2}\rfloor + n + 2}\\
        p_{k+1}^{-\lfloor \frac{k+2}{2}\rfloor + n +1} & e_{k+1}^{- \lfloor \frac{k+2}{2}\rfloor + n + 1} & d_{k+1}^{-\lfloor \frac{k+2}{2}\rfloor+n} 
        \end{Bmatrix}\\
        &&\times \frac{\theta(d_k^{- \lfloor \frac{k+1}{2}\rfloor + n + 1},e_{k+1}^{- \lfloor \frac{k+2}{2}\rfloor + n + 1},d_k^{-\lfloor \frac{k+1}{2}\rfloor + n + 2}}{\Delta_{d_k^{- \lfloor \frac{k+1}{2}\rfloor + n + 2}}}\\
        &&\times\mathcal O_{n-1}. 
    \end{aligned}
\end{eqnarray}
The proof of this fact can be found in Lemma~\ref{lem:bubble_rewriting}. Observe that the formula holds for $n\geq 4$, since this step does not appear in the cases $n=2, 3$, as a direct inspection reveals. Observe that properly speaking, the coefficients $p_{k+1}$ corresponding to $k=2n-5$ in the product above are identified with other $p$ coefficients through the Schur's Lemma (i.e. a Kroenecker's delta) applied when obtaining Equation~\eqref{eqn:half-bubble}. 

For simplicity of notation, we set $\Phi(\bar a, \bar b, \bar c, \bar d, \bar e, \bar f)$ to be the coefficient appearing in the RHS of Lemma~\ref{lem:bubble_rewriting}. If summation is to be taken over some of the indices, let us denote them as $\bar i$, then we indicate these indices  explicitly as $\Phi(\bar a, \bar b, \bar c, \bar d, \bar e, \bar f\ |\ \bar i)$. For short, in this situation, we also write $\Phi_{\bar i}$, when the labels are understood. 
We have 
\begin{eqnarray}
    \mathcal{HH}_{n+1} \circ_{\bar v} \mathcal O_n = \mathcal H_{n+1},
\end{eqnarray}
where $\bar v$ is the set of vertices as in Lemma~\ref{lem:half-bubble}.

To obtain the general evaluation of the spin network $\mathcal H_n$ for arbitrary $n$, we now proceed inductively by decomposing $\mathcal H_n$ into the composition of $\mathcal H_{n-1}$ and a term $\mathcal O_n$ whose evaluation can be obtained applying recoupling theory. Throughout, the labels $\bar a, \bar b, \bar c, \bar d, \bar e$ indicating the colorings assigned to the spin-network will follow the scheme described above. 

We are now in the position to relate the evaluation of the honeycomb spin-network $\mathcal H_{n+1}$ to the evaluation of $\mathcal H_n$, for any given configuration of the spin colors. First, we absorb all the coefficients $\Psi_{\bar i}$ and the extra factors coming from Lemma~\ref{lem:half-bubble} and Lemma~\ref{lem:bubble_rewriting} to get the new coefficients $\hat \Psi_{\bar i}$ and $\iota \hat \Psi_{\bar i}$. Observe, in fact, that apart from some pre-factors appearing in Lemma~\ref{lem:half-bubble}, all the coefficients are symmetric with respect to the involution $\iota$. We therefore use the symmetry to define the terms $\hat \Psi_{\bar i}$, and give a square root factor of the terms that are fixed by $\iota$. This preserves the symmetry between $\hat \Psi_{\bar i}$ and $\iota \hat \Psi_{\bar i}$.
We have
\begin{equation}
\begin{split}
    \mathcal H_{n+1}(\bar a, \bar b, \bar c, \bar d, \bar e) =& \sum_{\bar i} \hat \Psi(\bar a, \bar b, \bar c, \bar d, \bar e\ |\ \bar i) \iota \hat \Psi(\bar a, \bar b, \bar c, \bar d, \bar e\ | \bar i) \\ &\times\mathcal H_{n}(\bar a, \bar b, \bar c, \bar d, \bar e), 
\label{eqn:inductive}
\end{split}
\end{equation}
where $\mathcal H_{n}$ inherits the spin colors of $\mathcal H_{n+1}$. This important result is stated and proved in Theorem~\ref{thm:inductive_step}.

A fundamental computational/algorithmic issue that arises in the evaluation of $\mathcal H_n$ following Theorem~\ref{thm:inductive_step} regards the inductive determination of the new labels $\bar a', \bar b', \bar c', \bar d', \bar e'$. In fact, observe that while the labels in the bulk of the spin-network $\mathcal H_{n-1}$ obtained from the ``higher degree'' $\mathcal H$ remain the same in the inductive process outlined in the proof of Theorem~\ref{thm:inductive_step}, the same does not hold true for all the labels in the upper perimeter. In fact, as a consequence of the proof, there are $2n-3$ labels that we are going to sum over after applying recoupling an appropriate number of times. For instance, in the evaluation of $\mathcal H_2$, we sum on a single $i$, while in $\mathcal H_3$ we sum over $3$ and so on. These colorings we sum upon are then taken into account in the colorings of $\mathcal H_{n-1}$, and to concretely evaluate $\mathcal H_n$ (see appendix) one needs to iteratively take these colorings into account, and device a scheme for the substitution. As the edges where we sum the spin colors all lie in the upper semi-perimeter of $\mathcal H_{n-1}$ (along $\mathcal O_n$ in the decomposition of $\mathcal H_n$) following the proof of Theorem~\ref{thm:inductive_step}, this is not difficult to perform iteratively.

We find that the number of summation operations needed to evaluate $\mathcal H_n$ grows quadratically with $n$. More specifically, if $a_n$ denotes the number of summations at $n$, we have $a_n = a_{n-1} + 2n - 5$. This is a consequence of Equation~\ref{eqn:inductive} (i.e. Theorem~\ref{thm:inductive_step}) and it is proved in Corollary~\ref{cor:linear_sums} below. 

Another consequence of Equation~\ref{eqn:inductive} is that the evaluation of $\mathcal H_n(\bar a, \bar b, \bar c, \bar d, \bar e)$, with $n\geq 2$, is given by the formula 
\begin{eqnarray*}
\mathcal H_n(\bar a, \bar b, \bar c, \bar d, \bar e) = \sum_{k = 2}^n\Psi_{\bar i_0}\Phi_{\bar i'_0} \theta(c^2_1,e^2_0,b^2_0),
\end{eqnarray*}
where $\Psi_{\bar i_k}$ and $\Phi_{\bar i '_k}$ have been provided above and the index $k$ refers to the superscript of the indices of the spin colors $a, b, c, d, e$. This is shown in Corollary~\ref{cor:honeycomb_formula}. 

The simplicity of the formula given in Corollary~\ref{cor:honeycomb_formula} is not fully representative of the intrinsic complexity of it. In fact, the main issue in computing the evaluation of the Honeycomb network for an arbitrary $n$ is that the indices appearing in the summation symbol, which refer to the quantum $6j$ symbols in the $\Psi$ and $\Phi$ coefficients, are not explicitly given, and need to be considered carefully. In fact, at each step, the spin-colors for the level $n-1$ contain indices of summation from the previous step.


\section{Evaluation of $\mathcal H_n$}
\label{algo}

In this section we give the algorithm for the evaluation of the spin-network $\mathcal H_n$ using the steps described in the previous sections. Before giving the general procedure, we will consider an example in detail. We will compute the evaluation of $\mathcal H_n$ for arbitrary colors $a^i_k, b^i_k, c^i_, d^i_k, e^i_k$. The honeycomb $\mathcal H_3$ is the first of the $\mathcal H_n$ where the various steps of the algorithm are nontrivial, and it therefore shows the procedure, but with a complexity relatively small and still simple to perform by hand. The spin-network $\mathcal H_3$ with the labeling described above is shown in Figure~\ref{fig:H3}. Observe that some of the labels are merged into a single spin color. This is due to the fact that at a binary vertex, different colors would imply that the spin-network is trivial, and therefore it is meaningful to consider only the case when all the perimeter labels are grouped in a way that at binary vertices the incoming edges have the same spin color. Also, the composition of projectors at incident binary vertices squares to the identity, and several concatenated projectors result in a single projector. In other words we can consider these edges as a single ``smoothed'' edge. We specify also that the procedure given to pass from pixel space to spin-networks automatically implies that the spin colors are the same at these edges, and the spin-network is not trivial due to mismatches at the binary vertices. 

\begin{center}
\begin{figure}[htb]
\includegraphics[width=2in]{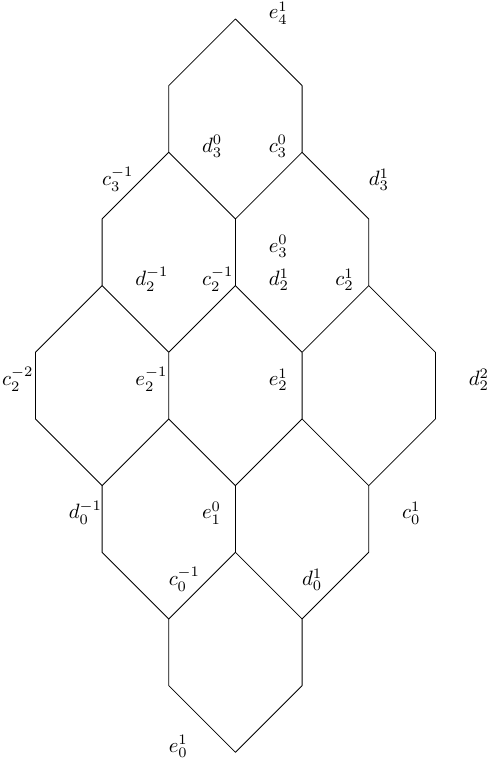}
\caption{The spin-network $\mathcal H_3$ with the labeling scheme adopted in this article. Some of the edges' labels are merged due to the fact that each edge is symmetrized through the Jones-Wenzl projector which, being a projector, is the identity when squared.}
\label{fig:H3}
\end{figure}
\end{center}

First, we apply Lemma~\ref{lem:bubble-move} to the top of the spin-network to obtain a factor of $\begin{Bmatrix}
    d_3^0 & c_3^{-1} & e_3^0\\
    d_3^1 & c_3^0 & e_4^1
\end{Bmatrix}
\cdot \Delta^{-1}_{e_3^0} \theta (d_3^0, c_3^0, e_3^0) 
$
multiplying the spin-network of Figure~\ref{fig:bubble_move}
\begin{center}
\begin{figure}[htb]
\includegraphics[width=2.5in]{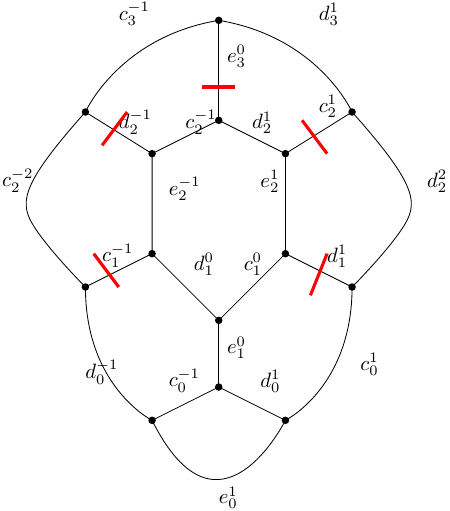}
\caption{First step of the algorithm applied to $\mathcal H_3$.}
\label{fig:bubble_move}
\end{figure}
\end{center}

Next, we apply the recoupling Theorem centered on the edges that have a perpendicular red marker. These recouplings can be applied in parallel, in the sense that they do not depend on each other, and the procedure can be performed simultaneously. Each recoupling now implies that a summation on compatible colors appears, along with a $6j$-symbol. The indices used for summation will be denoted by $p$, and we obtain a global coefficient 
\begin{equation}
\begin{split}
&\sum_{p_3^0} \begin{Bmatrix}
c_2^{-1} & c_3^{-1} & p_3^0\\
d_3^1 & d_2^1 & e_3^0
\end{Bmatrix}
\sum_{p_2^{-1}}
\begin{Bmatrix}
e_2^{-1} & c_2^{-2} & p_2^{-1}\\
c_3^{-1} & c_2^{-1} & d_2^{-1}
\end{Bmatrix}
\sum_{p_3^0} \begin{Bmatrix}
d_2^{1} & d_3^{1} & p_2^1\\
d_2^2 & e_2^1 & c_2^1
\end{Bmatrix}\\
&\times\sum_{p_3^0} \begin{Bmatrix}
d_0^{-1} & c_2^{-2} & p_1^{-1}\\
e_2^{-1} & d_1^0 & c_1^{-1}
\end{Bmatrix}
\sum_{p_3^0} \begin{Bmatrix}
c_1^{0} & e_2^{1} & p_1^1\\
d_2^2 & d_2^1 & e_3^0
\end{Bmatrix}\\
\end{split}
\end{equation}
with the resulting spin-network given in Figure~\ref{fig:crown_bubble}.
\begin{center}
\begin{figure}[htb]
\includegraphics[width=2.5in]{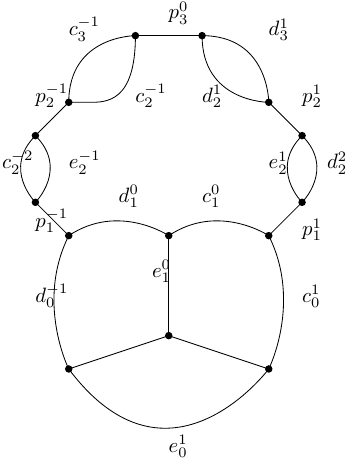}
\caption{Second step of the algorithm applied to $\mathcal H_3$, where we have applied recoupling to all red marked edges of Figure~\ref{fig:bubble_move}.}
\label{fig:crown_bubble}
\end{figure}
\end{center}
Now we can apply the diagrammatic Schur's Lemma (Lemma~7 in \cite{KL}) to all the bubbles appearing in Figure~\ref{fig:crown_bubble} and burst them all. This procedure introduces some $\theta$'s and quantum dimensions in the coefficients, but more importantly introduces Kronecker's deltas among the indices $p$'s. The coefficient multiplying every summand now becomes
\begin{equation}
\begin{split}
&\sum_{p_3^0} \begin{Bmatrix}
c_2^{-1} & c_3^{-1} & p_3^0\\
d_3^1 & d_2^1 & e_3^0
\end{Bmatrix}
\begin{Bmatrix}
e_2^{-1} & c_2^{-2} & p_3^0\\
c_3^{-1} & c_2^{-1} & d_2^{-1}
\end{Bmatrix}
\begin{Bmatrix}
d_2^{1} & d_3^{1} & p_3^0\\
d_2^2 & e_2^1 & c_2^1
\end{Bmatrix}\\
&\times\begin{Bmatrix}
d_0^{-1} & c_2^{-2} & p_3^0\\
e_2^{-1} & d_1^0 & c_1^{-1}
\end{Bmatrix}
\begin{Bmatrix}
c_1^{0} & e_2^{1} & p_3^0\\
d_2^2 & d_2^1 & e_3^0
\end{Bmatrix}\\
&\times
\frac{\theta(c_2^{-2}, e_2^{-1}, p_3^0)}{\Delta_{p_3^0}}
\frac{\theta(c_3^{-1}, c_2^{-1}, p_3^0)}{\Delta_{p_3^0}}
\frac{\theta(d_2^{1}, d_3^{1}, p_3^0)}{\Delta_{p_3^0}}
\frac{\theta(d_2^{2}, e_2^{1}, p_3^0)}{\Delta_{p_3^0}}
\end{split}
\end{equation}
and the spin-network we obtain (for each given configuration of spin-colors) is given by Figure~\ref{fig:pre_tetra}.
\begin{center}
\begin{figure}[htb]
\includegraphics[width=2.5in]{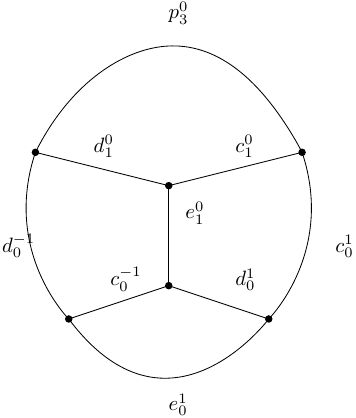}
\caption{Spin-network obtained from Figure~\ref{fig:crown_bubble} after bursting the bubbles through the diagrammatic Schur's Lemma.}
\label{fig:pre_tetra}
\end{figure}
\end{center}
One extra application of Lemma~\ref{lem:bubble-move} now allows us to obtain a sum (over compatible spin-colors) of terms that are proportional to tetrahedra, where the previous coefficients now get an extra factor of
$\begin{Bmatrix}
    d_1^0 & d_0^{-1} & e_1^0\\
    c_0^1 & c_1^0 & p_3^0
\end{Bmatrix}
\cdot \Delta^{-1}_{e_1^0} \theta (d_1^0, c_1^0, e_1^0). 
$
Since the evaluation of the tetrahedron is known (see Section~8.5 in \cite{KL}), the algorithm stops, and we can evaluate the original $\mathcal H_3$ through a sum over the compatible spin-color, evaluations of tetrahedra, and evaluations of $6j$-symbols and $\theta$-nets.   

The procedure just described for $\mathcal H_3$, exemplifies the whole theory in Section~\ref{sec:Honey} for the evaluation of $\mathcal H_n$, and gives a concrete realization of the results of Theorem~\ref{thm:inductive_step} to pass from $\mathcal H_3$ to $\mathcal H_2$ (which is a tetrahedron).

\begin{algorithm}[H]
\caption{General algorithm for the evaluation of $\mathcal H_n$.}
\label{algo:Evaluation}
\begin{algorithmic}[1]
\Require{$\mathcal H_n$ with given spin-colors}
\Comment{Initialization}
\Ensure{$\langle\mathcal H_n\rangle$} \Comment{Evaluation of $\mathcal H_n$}
\While{While $\mathcal H_n$ with $n\geq 3$}
\State{Apply Lemma~\ref{lem:bubble-move} to top of $\mathcal H_n$}
\State{Apply Lemma~\ref{lem:bubble-move} to use recoupling on all edges that connect crown to bulk}
\State{Remove bubbles through Schur's Lemma}
\State{Apply Lemma~\ref{lem:bubble-move} to edges connecting $\mathcal{HH}_{n-1}$ to $\mathcal{B}_n$}
\State{Apply Lemma~\ref{lem:bubble_rewriting} write $\mathcal {BO}_n$ in terms of $\mathcal O_n$}
\State{Apply Lemma~\ref{lemma:half_octopus} to obtain $\mathcal H_n$}
\EndWhile
\State{Perform sum over all compatible colors from the while}
\State{Evaluate the tetrahedra}
\end{algorithmic}
\end{algorithm}

\section{Computation of transition amplitudes}
\label{sec:comput}

To compute the transition amplitudes between two different hexagonal spin-networks $\mathcal H_1$ and $\mathcal H_2$, we compute the physical inner product  defined by Noui and Perez by means of a projector $P$ \cite{NP}. The definition of \cite{NP} was extended in \cite{CosmK} to the case of a projector where the quantum recoupling theory at non-classical $q$ is used. Physically, this corresponds to the case where the cosmological constant is nontrivial. We will refer to projector and physical product in the classical and quantum case interchangeably. 

A direct verification using the definitions found in \cite{CosmK} shows that the Haar integration (depicted as black boxes in \cite{NP}) satisfies the gauge fixing and summation identities found in the appendix of \cite{NP} in the quantum case as well. Using these two properties of the integration and the definition of the projector, we can reduce the computation of the transition amplitudes to evaluations as in Section~\ref{sec:Honey}. 

Let $\mathcal P$ denote the projector of \cite{NP} as well as the modified version of the quantum case (\cite{CosmK}). Then, the transition amplitude between two spin-networks $\mathcal H_n$ and $\mathcal H_n'$ of the same size, i.e. the physical inner product, is defined by the formula
$$
\langle \mathcal H_n | \mathcal H_n'\rangle_{\rm Phys} :=  \langle \mathcal H_n | \mathcal P | \mathcal H_n'\rangle,
$$
where $\langle \bullet | \bullet \rangle$ indicates the inner product defined via the Ashtekar-Lewandowski measure.  

It can be shown that Equation~\ref{eqn:inductive} suffices to evaluate transition amplitudes as follows. Then, the physical inner product between $\mathcal H_n$ and $\mathcal H_n'$ is given by 
\begin{eqnarray}
    \langle \mathcal H_n | \mathcal H_n'\rangle_{\rm Phys} = \overline{\langle \mathcal H_n\rangle} \langle \mathcal H_n'\rangle,
\end{eqnarray}
where $\langle \mathcal H_j\rangle$ indicates the evaluation computed in Section~\ref{sec:Honey} and the overbar denotes complex conjugation. This is proved in Lemma~\ref{lem:inner_product}, and the main step is to use Figure~\ref{fig:eliminate_bulk} to decouple the evaluation of the two spin-networks (see proof of Lemma~\ref{lem:inner_product} below).

    \begin{center}
        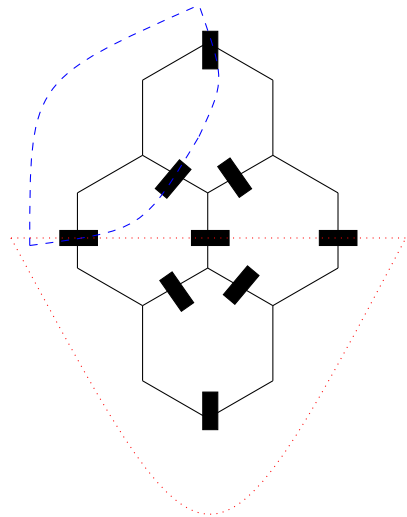
\begin{figure}[h!]
        \begin{tikzpicture}[rotate=30]
        \foreach \a in {0,120,-120} \draw (3*1,2*sin{60}*1) -- +(\a:1);
        \foreach \a in {0,120,-120} \draw (3*1,2*sin{60}*2) -- +(\a:1);
        \foreach \a in {0,120,-120} \draw (3*0+3*cos{60},2*sin{60}*1+sin{60}) -- +(\a:1);
        \foreach \a in {-120,120} \draw (3*1+3*cos{60},2*sin{60}*1+sin{60}) -- +(\a:1);
        \foreach \a in {0,-120} \draw (3*0+3*cos{60},2*sin{60}*2+sin{60}) -- +(\a:1);
        \foreach \a in {-120} \draw (3*1+3*cos{60},2*sin{60}*2+sin{60}) -- +(\a:1); 
        \foreach \a in {120} \draw (3*1+3*cos{60},2*sin{60}*2+sin{60}) -- +(\a:1);  
        \foreach \a in {0} \draw (3*1,2*sin{60}*3) -- +(\a:1);
        \foreach \a in {-120} \draw (3*1,2*sin{60}*3) -- +(\a:1);
        \foreach \a in {0} \draw (3*0+3*cos{60},2*sin{60}*0+sin{60}) -- +(\a:1);    
        \foreach \a in {120} \draw (3*0+3*cos{60},2*sin{60}*0+sin{60}) -- +(\a:1);
        \draw[fill=black,rotate=185] (-2.1,-2.1) rectangle (-2.3,-2.6);
        \draw[fill=black,rotate=60] (3.8,-0.65) rectangle (4,-1.15);
        \draw[fill=black,rotate=110] (1,-3.1) rectangle (1.2,-3.6);
        \draw[fill=black,rotate=110] (2.6,-3.6) rectangle (2.8,-4.1);
        \draw[fill=black,rotate=185] (-3.6,-2.9) rectangle (-3.8,-3.4);
        \draw[fill=black,rotate=60] (3.8,1.1) rectangle (4,0.6);
        \draw[fill=black,rotate=60] (3.8,-2.35) rectangle (4,-2.85);
        \draw[fill=black,rotate=150] (-1,-1.35) rectangle (-0.8,-1.85);
        \draw[fill=black,rotate=150] (-1,-6.15) rectangle (-0.8,-6.65);
        %
        %
        \draw[dashed,rotate=-30,color=blue] (-1.5,3.8) ..controls (0,4).. (0.75,5.25);
        \draw[dashed,rotate=-30,color=blue] (0.75,5.25) ..controls (1.1,6).. (0.75,7);
        \draw[dashed,rotate=-30,color=blue] (-1.5,3.8) ..controls (-1.5,6)..(0.75,7);
        \draw[dotted,rotate=-30,color=red] (-1.75,3.9) -- (3.5,3.9);
        \draw[dotted,rotate=-30,color=red] (-1.75,3.9) ..controls(0.9,-1).. (3.5,3.9);
        \end{tikzpicture}
            \caption{Elimination of Haar integration from the bulk. The blue dashed line shows the elimination of diagonal Haar boxes, while the dotted red line shows the elimination of the horizontal Haar box.}
            \label{fig:eliminate_bulk}
            \end{figure}
    \end{center}

    \section{Phase-Space Properties}\label{phase-space}

In this Section we explore the \emph{phase space} for the value of the Perez-Noui projector, relatively to two different sizes of the hexagonal grid, i.e. with $N=2,3$. Furthermore, we set the value of $q$ to be the \emph{classical} one, with $q=-1$. In order to deal with a finite number of coloring configurations for the hexagonal lattices, we need to set bounds on the possible compatible choices for each edge, i.e. we need to impose a minimum $c_m$ and a maximum $c_M$ color value, and enumerate all the possible coloring configurations in that range, with some constraint coming from the coloring procedure. This is a rather complex combinatorial problem: a first straightforward approach would be to randomly draw colors for each edge and imposing the compatibility conditions at the vertices, with the drawback of searching among $(c_M - c_m)^{N_{e}}$ combinations, with $N_e$ the total number of edges, among which only a very small fraction actually yields compatible colorings. 

As it appears, the main problem is to find a procedure that automatically yields compatible color configurations. The solution we put forward is to color the graph using its cycles as the fundamental units: assuming one finds all possible graph cycles $\{\gamma^{(N)}_i\}$ (i.e. sequences of edges that form close loops) for the $N\times N$ hexagonal lattice, then one can build compatible colorings configurations in the range $[c_m, c_M]$ by increasing by one the color of each edge belonging to a given cycle $\gamma^{(N)}_i$, with the possibility of increasing multiple times the value of the colors of the edges belonging to the any given cycle. This is a non-local construction of the colorings that automatically assures the compatibility of each configuration.

Hence, after enumerating all the cycles, one can build all the possible configurations of maximum cycle color $c_M=1$ simply by coloring one cycle per configuration; then one can build all configurations of maximum cycle color $c_M=2$ by coloring all possible combinations of pairs of cycles, including choosing the same cycle twice, and so on.

In this way, we introduce a possible parametrization of the phase space of all the (infinite) compatible colorings that is based on coloring cycles in order to assure the compatibility of each configuration. As a final remark, it is important to consider that finding all possible cycles of the hexagonal graph is again a non-trivial combinatorial problem for which we have developed our own strategy, which will be described in a separate work.

Let us now discuss the results for the projector values among any couple of configurations, in relation to a given range of cycles-colorings, for $N=2,3$. As it turns out, it is not possible to store in memory the results for $N=4$: the number of cycles in this case is $N^{(4)}_c=18370$, which would all yield the same evaluation, while the number of configurations for all pairs of cycles is $N^{(4)}_{c_M=2}=168737635$, which does not allow to compute all the possible transition values and store them on RAM considering $32$-bits precision. Hence, we choose to consider $N=2$ with $c_m=0$ and $c_M=6$, and $N=3$ with $c_m=0$ and $c_M=2$, yielding a total number of transition values of $N^{(3)}_{c_m=0,c_M=6}=1502260081$ and $N^{(3)}_{c_m=0,c_M=2}=1569744400$. 

We start by associating an integer index $i$ to each coloring configuration $|\mathcal{H}_n^{(i)}\rangle$, so that the transition matrix $\mathcal{A}$ reads
\begin{equation}
    \mathcal{A}_{ij} = |\langle \mathcal{H}_n^{(i)}|\mathcal{H}_n^{(j)}\rangle |_{\text{Norm}}^2 = \frac{|\overline{\langle\mathcal{H}_n^{(i)}\rangle}\langle\mathcal{H}_n^{(j)}\rangle|^2}{\max{\{|\langle\mathcal{H}_n^{(i)}\rangle|^4, |\langle\mathcal{H}_n^{(j)}\rangle|^4\}}}\,,
\end{equation}
which is such that the diagonal part is normalized to unity. Any random labelling of the coloring states $|\mathcal{H}_n\rangle$ would not yield any apparent structure in the density matrix, hence we decided to rank each state by means of the sum of all the transition probability values between the given state and all the others, i.e.
\begin{equation}
    \mathcal{S}_i=\sum_j \mathcal{A}_{ij}.
\end{equation}
As shown in Figs.~\ref{fig:n2qm1} and~\ref{fig:n3qm1}, if one reorders the labeling according to increasing values of $S_i$, one can use the new ranked indices $\{i_\text{R}\}$ and represent the ranked transition matrix, denoted as $\mathcal{A}_{i_\text{R}j_\text{R}}$. One can then see that the values are automatically structured in a block diagonal form, where different states cluster in what we refer to as \emph{classes}: within one class each state is equivalent to the others, in the sense that the transition probability is unity.

\begin{figure}[t!]
    \centering
    \includegraphics[scale=0.57]{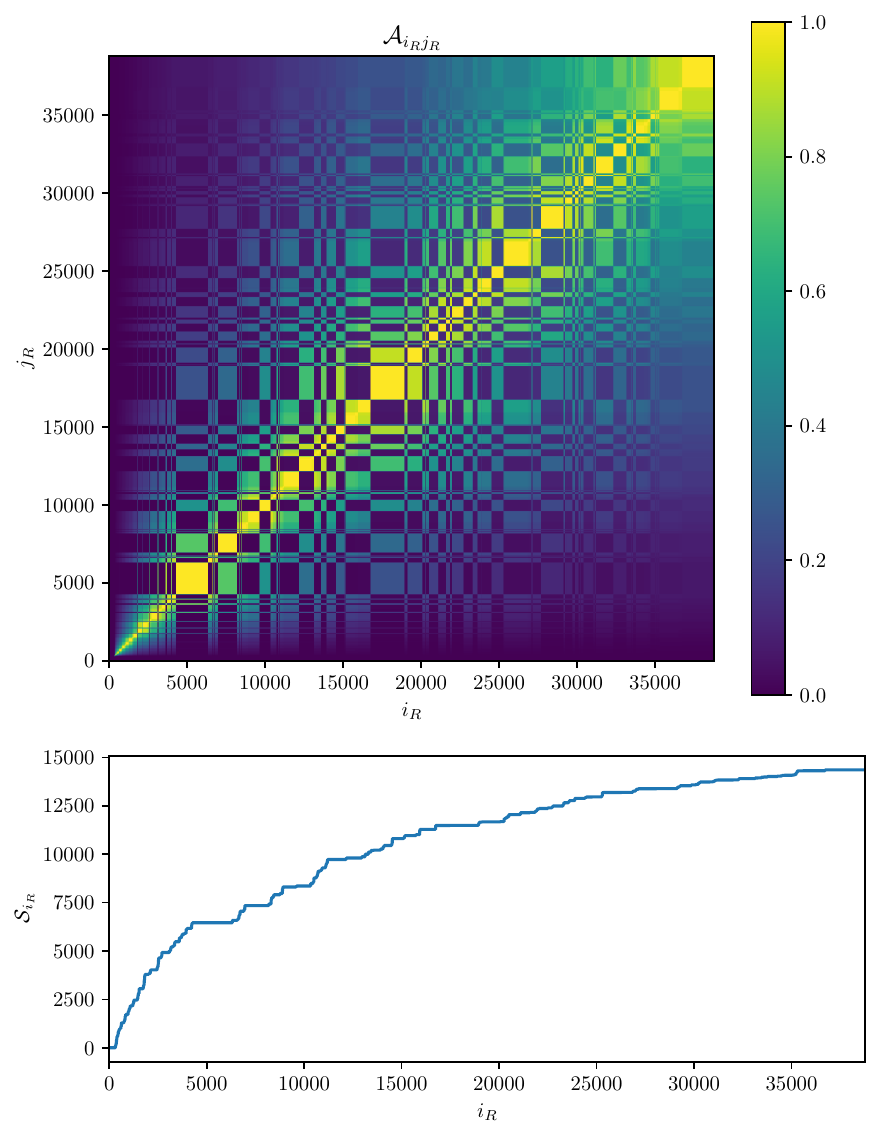}
    \caption{Transition matrix for $N=2$ honeycomb lattice with $q=-1$, maximum cycle-coloring value $c_M=6$.}
    \label{fig:n2qm1}
\end{figure}

\begin{figure}[t!]
    \centering
    \includegraphics[scale=0.57]{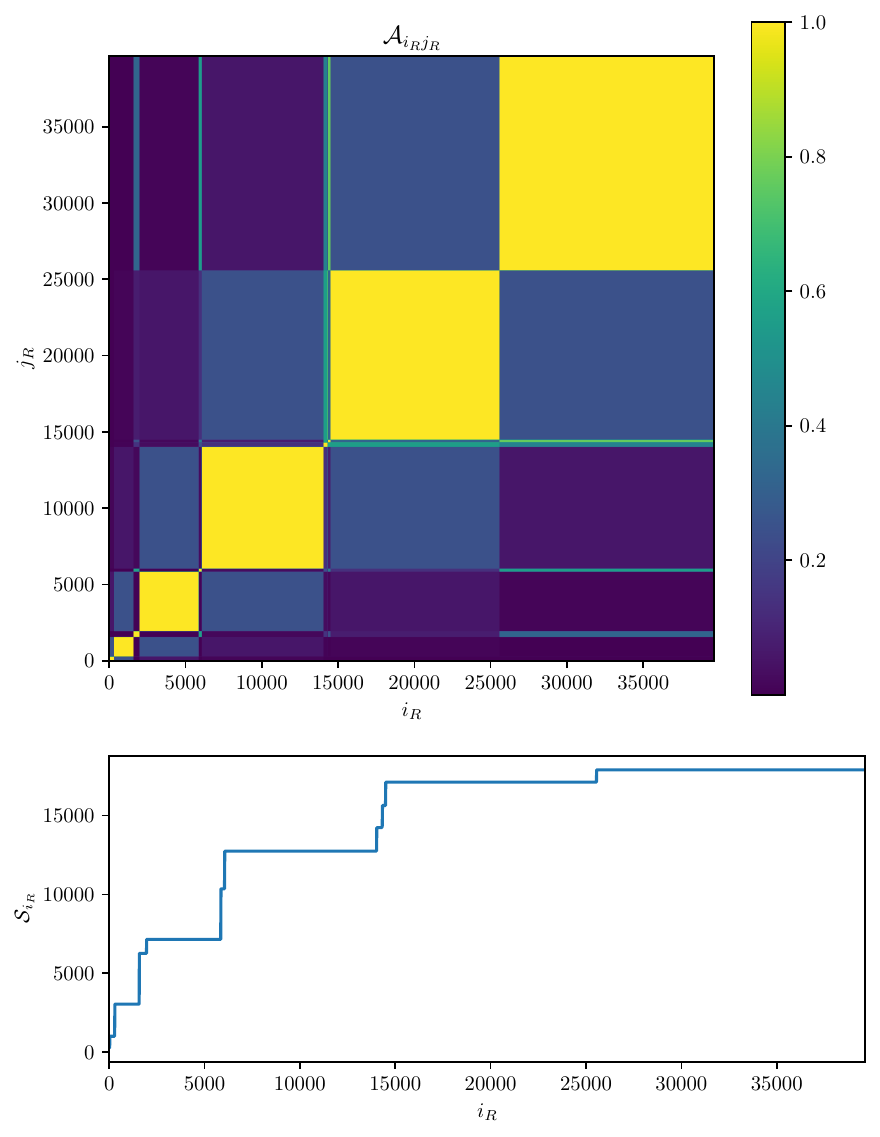}
    \caption{Transition matrix for $N=3$ honeycomb lattice with $q=-1$, maximum cycle-coloring value $c_M=2$.}
    \label{fig:n3qm1}
\end{figure}

Another remarkable property is that all the elements of a class also share the same value of the total sum $\mathcal{S}_i$. This feature provides an additional property of these classes: each element belonging to a class has the same global scalar product with all the other elements within the configuration space. This structure is showing how the Perez-Noui projector can be used to distinguish one class of elements from the other, without any prior information, i.e. training, with a structure that spontaneously emerges when considering a simple ranking of the states. In other words, this results is providing direct evidence for the ideas discussed in~\cite{marciano2022deep}, for which one expects DNNs to emerge as a semi-classical limit of TQNNs. In other words, the block-diagonal part of the transition matrix $\mathcal{A}$ is a way of representing the saddle point that would be found by training a classifier DNNs on the portion of the configuration space we study here. These results seem very promising for using TQNNs as an image classifier.

\section{Relation with the Ising model}
\label{Isi}
 
DNNs present many affinities with statistical models. Specifically, DNNs' architectures can be addressed from the perspective of statistical physics and Gibbs distributions. An area of research that was very active in the 80’s was the one hinging on the implementation of spin-glass models to unveil the way neural networks operate. A flourishing statistical approach is also represented by the so-called Boltzmann machines, networks of neuron-like units symmetrically connected, with neurons can be switched on or off according to a stochastic dynamics. Their learning algorithm \cite{HS} allows to achieve complex pattern recognition tasks by adopting a supervised approach. Boltzmann machines emerge as stochastic recurrent neural networks, which have been cast in statistical physics as the disordered versions of the Ising model \cite{Ising}, i.e. the Sherrington-Kirkpatrick model \cite{SK}. \\

In particular, the generalisation ability was one of the battle field of these investigations inspired by the statistical analysis of phase transitions. Quantum fluctuations can be rephrased as statistical fluctuations, by means of a standard Wick rotation. This latter transforms the partition function of any quantum theory in the equivalent partition function in statistical mechanics provided with the Gibbs-ensemble measure, namely the negative exponential of the Hamiltonian of the system. On the other hand, the connectivity does naturally enter inside the definition of the semi-classical limit of the QNN/TQNN states through the concept of coarse-graining. Borrowing an intuition proper of statistical mechanics, we may think that blocking and coarse-graining procedures, directly applied at the quantum level on the TQNN states, individuate a class of effective TQNN states that are supported on graphs characterised by a lower topological connectivity, and thus by a lower capacity --- we call these states statistical TQNN (STQNN). More concretely, from an operative point of view, the blocking and the coarse-graining procedures are defined in terms of the ability to carry out measurements.\\ 

\section{Conclusions and outlooks}
\label{conclu}

The enhancement of computational methods is the omnipresent driving factor of today’s scientific panorama. The advancement of technological instrumentation has allowed researchers in any field to gather increasingly more data about virtually any aspect of natural science. Nonetheless, advancements in computational ability, with eventual breakthrough, are still required, and probably even more needed that in the past.

Quantum computing may represent a milestone along this trajectory. It may pave the way to a shift of perspective in computational methods, with outputs that are qualitatively different and not comparable with classical computing. Quantum computing may furthermore enable to process data in quantum machines, including quantum computers, exploiting the quantum structures of matter. \\

In this article we have delved into the evaluation of spin-networks of hexagonal shape and arbitrary size. We have hence related these objects to the pixel space of images, in order to apply the new tools provided by topological quantum neural networks (TQNNs). We have then constructed an algorithm for the evaluation of the Perez-Noui projector on $SU(2)$ \cite{NP}, and extended this result to $SU_q(2)$ \cite{CosmK}.
\\

Some aspects of our construction will deserve more detailed investigations in the future. The link between ``local’’ features and ``global’’ ones is among these, and appears of particular interest.

The squared norm of the normalized physical scalar product between two different states $\mathcal A_{nn'} = |\langle \hat{\mathcal H}_n | \hat{\mathcal H'}_n\rangle|^2$ can be used to rank the states as follows: fix the state $| \hat{\mathcal H'}_n\rangle$ and compute the partial sum $\mathcal S_{n'} = \sum_n \mathcal A_{nn'}$; as it happens the value of $\mathcal S_{n'}$ can be used to rank each state $| \hat{\mathcal H'}_n\rangle$. At the end of the ranking procedure one finds that the ranked matrix $\bar{\mathcal A}_{nn'}$ has a block diagonal structure where the blocks are all related to transitions $|\langle \hat{\mathcal H}_n | \hat{\mathcal H'}_n\rangle|^2 = 1$. It also happens that each block is associated to a unique value of the partial sum $\mathcal S_{n'}$.

Hence, the states belonging to the blocks display two fundamental property: a ``local'' property, i.e. the fact that each state has a scalar product equal to one over any other state belonging to the same block; a ``global'' property, i.e. that all the states belonging to a block yield the same value for the partial sum $\mathcal S_{n'}$. This is a remarkable property that links a local feature to a global one. It is possible to associate each of the diagonal blocks to a ``class'' that, upon visual inspection, seems to yield reasonably distinguishable spin networks in terms of the coloring.\\

The origin of the classification mechanism also deserves more detailed analyses. If one assumes that the overall set of all possible transitions, computed using the Perez-Noui projector, allows to compute the Turaev-Viro invariant, then it might be possible that the partial sum $\mathcal S_{n'}$ is related to the Reshetikhin-Turaev invariant. If this is the case, then each diagonal block might be related to a different value of the Reshetikhin-Turaev invariant thus providing a mathematical foundation for the mechanism that is yielding the classification we observe in the ranked transition matrix $\bar{\mathcal A}_{nn'}$.

Assuming that the Turaev-Viro invariant can be computed from the transition matrix ${\mathcal A}_{nn'}$, then the diagonal blocks in $\bar{\mathcal A}_{nn'}$ might represent the saddle point of the Turaev-Viro evaluation, if considering the TV as composed by a sum of the exponential of the values of $\bar{\mathcal A}_{nn’}$.\\

In conclusion, the intrinsic quantumness of the TQNN framework \cite{TQNN,Fields:2022bti}, in which the dynamical evolution of the boundary states (input/output data) is attained through the sum over an infinite amount of intermediate virtual states (filters/hidden layers), has been realised here by applying the physical projectors to the spin-network states. The quantumness that is intrinsic in this proposed new framework allows us to consider a sum over infinite (virtual) hidden layers, being conjectured at the same time to avoid the issues of redundancy and overfitting \cite{marciano2022deep}. This instantiates novel (quantum) algorithms, the effectiveness and accuracy of which we will have to continue  testing, investigating the amount of computational time TQNNs spent in comparison with classical counterparts, such as deep neural networks (DNNs), and delving into the material implementations that exploit topological condensed matter structures described in terms of string-nets \cite{2003PhRvB..67x5316L, Levin:2004mi}. All results can be independently reproduced through the ``idea.deploy" framework~\href{https://github.com/lullimat/idea.deploy}{\nolinkurl{https://github.com/lullimat/idea.deploy}}

\appendix 

\section{Proofs of the results}

In this appendix we collect the main results (and their proofs) used in the article to obtain the algorithm. 

\begin{lemma}\label{lem:composition}
    It holds that $\mathcal H_{n+1} = \mathcal H_n \circ_{\bar v} \mathcal O_n$ for every  $n\in \mathbb N$, and for some choice of vertices $\bar v$ in $\mathcal H_n$. 
\end{lemma}
\begin{proof}
    The proof is by induction on $n$, and it does not depend on the colorings of the spin-networks, so that we can omit keeping track of the spin colors, but we can just consider the underlying graphs. The base of induction holds true, since for $n=1$ the graph $\mathcal H_n$ is just a single hexagon cell, and $\mathcal H_2$ is obtained by attaching $\mathcal O_1$ on the three top vertices of the hexagon cell. Suppose now that the result has been proved for some $k>1$, and let us consider $\mathcal H_{k+1}$. in the graph of $\mathcal H_{k+1}$ we can isolate a top layer of the graph, where we imagine of cutting the edges that connect the outer perimeter to the inner vertices of $\mathcal H_{k+1}$. This leaves a graph $\mathcal H_k$ and detaches an open-edge graph that is readily identified with a copy of the graph $\mathcal O_k$. We observe that in this step it might be necessary to eliminate extra vertices inside the edges of the detached graph. This is indeed possible since a binary vertex can be eliminated, and the symmetrizers that label the two edges are compacted into one, using idempotency of the Jones-Wenzl symmetrizer. 
\end{proof}

\begin{widetext}
\begin{lemma}\label{lem:bubble-move}
    The following equality holds for all choices of compatible spin colors $a, b, c, d, e, f$:
    \begin{equation*}
        \begin{tikzpicture}[baseline={([yshift=0.2cm]current bounding box.center)},vertex/.style={anchor=base,
        circle,fill=black!25,minimum size=18pt,inner sep=2pt}]
        \draw (-1,0) -- (0,0);
        \draw (0,0) .. controls (1,0.5).. (2,0);
        \draw (0,0) .. controls (1,-0.5).. (2,0);
        \draw (2,0) -- (3,0);
        \draw (1,-0.4) -- (1,-1.4);
        \draw[fill=black] (1,-0.375) circle (2pt);
        \draw[fill=black] (0,0) circle (2pt);
        \draw[fill=black] (2,0) circle (2pt);
        
        \node (a) at (-0.5,0.35) {$b$};
        \node (a) at (2.5,0.35) {$c$};
        \node (a) at (1.7,-0.5) {$d$};
        \node (a) at (0.3,-0.5) {$a$};
        \node (a) at (1,0.6) {$e$};
        \node (a) at (1.3,-1) {$f$};
        \end{tikzpicture}
        =
        \begin{Bmatrix}
        a & b & f\\
        c & d & e 
        \end{Bmatrix}
        \Delta_f^{-1}
        \theta(a,d,f)
        \ \
        \begin{tikzpicture}[baseline={([yshift=0.1cm]current bounding box.center)},vertex/.style={anchor=base,
        circle,fill=black!25,minimum size=18pt,inner sep=2pt}]
        \draw (-2,0) -- (0,0) -- (2,0);
        \draw (0,0) -- (0,-1);

        \draw[fill=black] (0,0) circle (2pt);
        \node (a) at (-1,0.35) {$b$};
        \node (a) at (1,0.35) {$c$};
        \node (a) at (0.2,-0.7) {$f$};
        \end{tikzpicture}\ \ .
    \end{equation*}
\end{lemma}
\begin{proof}
    Applying recoupling to the edge $e$ we obtain the equality 
        
        \begin{equation*}
        \begin{tikzpicture}[baseline={([yshift=0.2cm]current bounding box.center)},vertex/.style={anchor=base,
        circle,fill=black!25,minimum size=18pt,inner sep=2pt}]
        \draw (-1,0) -- (0,0);
        \draw (0,0) .. controls (1,0.5).. (2,0);
        \draw (0,0) .. controls (1,-0.5).. (2,0);
        \draw (2,0) -- (3,0);
        \draw (1,-0.4) -- (1,-1.4);
        \draw[fill=black] (1,-0.375) circle (2pt);
        \draw[fill=black] (0,0) circle (2pt);
        \draw[fill=black] (2,0) circle (2pt);
        
        \node (a) at (-0.5,0.35) {$b$};
        \node (a) at (2.5,0.35) {$c$};
        \node (a) at (1.7,-0.5) {$d$};
        \node (a) at (0.3,-0.5) {$a$};
        \node (a) at (1,0.6) {$e$};
        \node (a) at (1.3,-1) {$f$};
        \end{tikzpicture}
        =
        \sum_i
        \begin{Bmatrix}
        a & b & i\\
        c & d & e 
        \end{Bmatrix}
        \ \ 
        \begin{tikzpicture}[baseline={([yshift=0.3cm]current bounding box.center)},vertex/.style={anchor=base,
        circle,fill=black!25,minimum size=18pt,inner sep=2pt}]
        \draw (-2,0) -- (0,0) -- (2,0);
        \draw (0,0) -- (0,-0.5);
        \draw (0,-0.5) .. controls (0.3,-0.75) .. (0,-1);
        \draw (0,-0.5) .. controls (-0.3,-0.75) .. (0,-1);
        \draw (0,-1) -- (0,-1.5);

        \draw[fill=black] (0,0) circle (2pt);
        \draw[fill=black] (0,-0.5) circle (1.5pt);
        \draw[fill=black] (0,-1) circle (1.5pt);
        
        \node (a) at (-1,0.35) {$b$};
        \node (a) at (1,0.35) {$c$};
        \node (a) at (0.2,-0.3) {$i$};
        \node (a) at (0.5,-0.75) {$d$};
        \node (a) at (-0.5,-0.75) {$a$};
        \node (a) at (0.3,-1.25) {$f$};
        \end{tikzpicture}\ \ .
    \end{equation*}
    Now, applying Lemma~7 of \cite{KL} (i.e. the diagrammatic Schur's Lemma) we find that the only term in the sum that is not trivial is the one corresponding to $i=f$, and moreover the previous equation becomes
    
    \begin{equation*}
        \begin{tikzpicture}[baseline={([yshift=0.2cm]current bounding box.center)},vertex/.style={anchor=base,
        circle,fill=black!25,minimum size=18pt,inner sep=2pt}]
        \draw (-1,0) -- (0,0);
        \draw (0,0) .. controls (1,0.5).. (2,0);
        \draw (0,0) .. controls (1,-0.5).. (2,0);
        \draw (2,0) -- (3,0);
        \draw (1,-0.4) -- (1,-1.4);
        \draw[fill=black] (1,-0.375) circle (2pt);
        \draw[fill=black] (0,0) circle (2pt);
        \draw[fill=black] (2,0) circle (2pt);
        
        \node (a) at (-0.5,0.35) {$b$};
        \node (a) at (2.5,0.35) {$c$};
        \node (a) at (1.7,-0.5) {$d$};
        \node (a) at (0.3,-0.5) {$a$};
        \node (a) at (1,0.6) {$e$};
        \node (a) at (1.3,-1) {$f$};
        \end{tikzpicture}
        =
        \begin{Bmatrix}
        a & b & f\\
        c & d & e 
        \end{Bmatrix}
        \Delta_f^{-1} \theta(a,d,f)
        \ \ 
        \begin{tikzpicture}[baseline={([yshift=0.1cm]current bounding box.center)},vertex/.style={anchor=base,
        circle,fill=black!25,minimum size=18pt,inner sep=2pt}]
        \draw (-2,0) -- (0,0) -- (2,0);
        \draw (0,0) -- (0,-1);

        \draw[fill=black] (0,0) circle (2pt);
        \node (a) at (-1,0.35) {$b$};
        \node (a) at (1,0.35) {$c$};
        \node (a) at (0.2,-0.7) {$f$};
        \end{tikzpicture}\ \ ,
    \end{equation*}
    where $\theta(a,d,f)$ denotes the value of the $\theta$-net
    \begin{center}
    \begin{tikzpicture}[baseline={([yshift=-0.1cm]current bounding box.center)},vertex/.style={anchor=base,
        circle,fill=black!25,minimum size=18pt,inner sep=2pt}]
    \draw (0,0) .. controls (1,1) .. (2,0);
    \draw (0,0) -- (2,0);
    \draw (0,0) .. controls (1,-1) .. (2,0);
    
    \node (a) at (1,0.35) {$f$};
    \node (a) at (1,1) {$a$};
    \node (a) at (1,-1) {$d$};
    \end{tikzpicture}\ \ .
    \end{center}
    The evaluation of the latter $\theta$-net cancels out with the of the renormalizations (see Appendix A of \cite{Rovelli}) $\sqrt{\theta(a,d,f)}$ of the two $3$-vertices $(a,d,f)$ and $(a,d,i)$, completing the proof.
\end{proof}
\end{widetext}

\begin{lemma}\label{lem:half-bubble}
    Let $\mathcal H_n(\bar a,\bar b,\bar c, \bar d, \bar e)$ be a honeycomb spin-network with the labeling scheme described above. Then we have 
    \begin{equation}
    \begin{split}
    \mathcal H_{n+1}&(\bar a,\bar b, \bar c, \bar d, \bar e) 
    =
    \Delta_{e_{n+1}^0}\theta(d_{n+1}^{0}, c_{n+1}^{0},e_{n+1}^0)\\
    &\times \frac{\theta(i_{n+1}^0,c_{n+1}^{-1},c_n^{-1})}{\Delta_{i_{n+1}^0}}\frac{\theta(i_{n+1}^0,d_{n+1}^{1},d_n^{1})}{\Delta_{i_{n+1}^0}}\\
    &\frac{\theta(i_{\lfloor \frac{n}{2}\rfloor}^{-\lfloor \frac{n}{2}\rfloor},c_{\lfloor \frac{n}{2}\rfloor+1}^{-\lfloor \frac{n}{2}\rfloor},e_{\lfloor \frac{n}{2}\rfloor+1}^{-\lfloor \frac{n}{2}\rfloor+1})}{\Delta_{i_{\lfloor \frac{n}{2}\rfloor}^{-\lfloor \frac{n}{2}\rfloor}}}
    \frac{\theta(i_{\lfloor \frac{n}{2}\rfloor}^{\lfloor \frac{n}{2}\rfloor},d_{\lfloor \frac{n}{2}\rfloor+1}^{\lfloor \frac{n}{2}\rfloor},e_{\lfloor \frac{n}{2}\rfloor+1}^{\lfloor \frac{n}{2}\rfloor-1})}{\Delta_{i_{\lfloor \frac{n}{2}\rfloor}^{\lfloor \frac{n}{2}\rfloor}}}\\
    &\times\begin{Bmatrix}
      d_{n+1}^{0}  & c_{n+1}^{-1}  & e_{n+1}^0 \\
      d_{n+1}^1  & c_{n+1}^{0}  & e_{n+2}^0
    \end{Bmatrix}\\
    &\times \mathcal {HH}_n(\bar a,\bar b,\bar c, \bar d, \bar e) \circ_{\bar v} \mathcal {BO}_n(\bar a,\bar b,\bar c, \bar d, \bar e)\\    
    &\times \sum_{\bar i}\Psi(\bar a,\bar b,\bar c, \bar d, \bar e\ |\ \bar i)\iota(\Psi(\bar a,\bar b,\bar c, \bar d, \bar e\ |\ \bar i)),\\
\end{split}    
    \end{equation}
    where $\Psi$ and $\iota$ were defined above, and the fractions appear only when $n>2$. 
\end{lemma}
\begin{proof}
    We proceed by using Lemma~\ref{lem:composition}, and recoupling theory. First, let us consider the simpler case $n=2$, which is verified as follows. We write $\mathcal H_2(\bar a, \bar b, \bar c, \bar d, \bar e) = \mathcal H_1(\bar a, \bar b, \bar c, \bar d, \bar e)\circ_{\bar v} \mathcal{O}_2(\bar a, \bar b, \bar c, \bar d, \bar e)$. Let us omit the labels $\bar a, \bar b, \bar c, \bar d, \bar e$ for simplicity. Then, we apply Lemma~\ref{lem:bubble-move} on the central edge right above the hexagonal cell of $\mathcal H_1$, as given in the decomposition of $\mathcal H_2$ above. The resulting spin-network is, with complex $6j$ factor multiplying it, $\mathcal{HH}_1\circ_{\bar v}\mathcal{BO}_1$, where $\bar v$ consists of the two vertices on the sides of the hexagonal cell $\mathcal H_1$. The complex factor appearing in the sum is the $6j$-symbol determined by Lemma~\ref{lem:bubble-move}. In this case there is a single $6j$, which is seen directly to coincide with the first factor in the formula in the statement of the lemma. The terms containing $\Psi$ and the fractions containing $\theta$ and $\Delta$ are not present in this case. 
    The case for arbitrary $n$ is similar, and it only requires more applications of the recoupling theorem. More specifically, we apply Lemma~\ref{lem:bubble-move} to the top of the spin network. This produces the factor
     $$
    \Delta_{e_{n+1}^0}\theta(d_{n+1}^{0}, c_{n+1}^{0},e_{n+1}^0)
    \begin{Bmatrix}
      d_{n+1}^{0}  & c_{n+1}^{-1}  & e_{n+1}^0 \\
      d_{n+1}^1  & c_{n+1}^{0}  & e_{n+2}^0
    \end{Bmatrix}
    $$
    which is the prefactor appearing in the statement. Then, we apply recoupling to the edges that are used to connect $\mathcal O_n$ to $\mathcal H_n$ in the decomposition $\mathcal H_{n+1} = \mathcal H_n \circ_{\bar v} \mathcal O_n$, along with the bottom edges of the most lateral hexagons. For each coloring, we now have to consider the coefficients appearing at each application of the recoupling theorem. 
    Now, proceeding along the left side of the graph supporting $\mathcal O_n$, we encounter the recoupling of edges $d_{k}^{-\lfloor \frac{n+2}{2}\rfloor + \lfloor \frac{k+1}{2}}$, while going in the opposite direction gives the recoupling on $c_{k}^{\lfloor \frac{n+2}{2}\rfloor - \lfloor \frac{k+1}{2}}$. This gives rise to the $6j$-symbols that constitute the terms indexed by $k$ appearing in the product that defines $\Psi$ and $\iota \Psi$, where one needs to sum over all the compatible $i$, with respect to the other entries of the $6j$-symbol. Finally, on the bottom edges of the equatorial belt of hexagons in the copy of $\mathcal H_n$ found inside of $\mathcal H_{n+1}$ we get recoupling on $c_{\lfloor \frac{n}{2}\rfloor + 1}^{-\lfloor \frac{n}{2}\rfloor-1}$ and $d_{\lfloor \frac{n}{2}\rfloor + 1}^{\lfloor \frac{n}{2}\rfloor+1}$, which gives rise to the last two factors in the definition of $\Psi$ and $\iota \Psi$. At this point we have a decomposition of the geometric support of the spin-network as $\mathcal {H}_n\circ \mathcal{BO}_n$ with extra four bubbles. Using Lemma~7 in \cite{KL} to burst the bubbles we obtain $\mathcal {H}_n\circ \mathcal{BO}_n$, and the remaining factors that consist of the fractions in the statement of the lemma. This completes the proof. 
\end{proof}

\begin{lemma}\label{lem:bubble_rewriting}
We have, for any $n\geq 4$, the equality
\begin{eqnarray*}
    \begin{aligned}
    \lefteqn{\mathcal {BO}_n(\bar a, \bar b, \bar c, \bar d, \bar e, \bar f)}\\ &=&     
        \prod_{n-1 \leq k \leq 2n - 5}
        \begin{Bmatrix}
        c_k^{\lfloor \frac{k+1}{2}\rfloor-n - 1} & p_k^{\lfloor \frac{k+1}{2}\rfloor-n - 1} & d_k^{\lfloor \frac{k+1}{2}\rfloor-n - 2}\\
        p_{k+1}^{\lfloor \frac{k+2}{2}\rfloor-n -1} & e_{k+1}^{\lfloor \frac{k+2}{2}\rfloor-n -1} & c_{k+1}^{\lfloor \frac{k+2}{2}\rfloor-n} 
        \end{Bmatrix}\\
        &&\times \frac{\theta(c_k^{\lfloor \frac{k+1}{2}\rfloor-n - 1},e_{k+1}^{\lfloor \frac{k+2}{2}\rfloor-n -1},d_k^{\lfloor \frac{k+1}{2}\rfloor-n - 2}}{\Delta_{d_k^{\lfloor \frac{k+1}{2}\rfloor-n - 2}}}\\
        &&\times 
        \begin{Bmatrix}
        d_k^{-\lfloor \frac{k+1}{2}\rfloor + n + 1} & p_k^{-\lfloor \frac{k+1}{2}\rfloor + n + 1} & c_k^{-\lfloor \frac{k+1}{2}\rfloor + n + 2}\\
        p_{k+1}^{-\lfloor \frac{k+2}{2}\rfloor + n +1} & e_{k+1}^{- \lfloor \frac{k+2}{2}\rfloor + n + 1} & d_{k+1}^{-\lfloor \frac{k+2}{2}\rfloor+n} 
        \end{Bmatrix}\\
        &&\times \frac{\theta(d_k^{- \lfloor \frac{k+1}{2}\rfloor + n + 1},e_{k+1}^{- \lfloor \frac{k+2}{2}\rfloor + n + 1},d_k^{-\lfloor \frac{k+1}{2}\rfloor + n + 2}}{\Delta_{d_k^{- \lfloor \frac{k+1}{2}\rfloor + n + 2}}}\\
        &&\times\mathcal O_{n-1}. 
    \end{aligned}
\end{eqnarray*}
\end{lemma}
\begin{proof}
  This is an application of the bubble move of Lemma~\ref{lem:bubble-move} to each bubble of $\mathcal {BO}_n$.
\end{proof}

\begin{lemma}\label{lemma:half_octopus}
We have 
$$
\mathcal{HH}_{n+1} \circ_{\bar v} \mathcal O_n = \mathcal H_{n+1},
$$
where $\bar v$ is the set of vertices as in Lemma~\ref{lem:half-bubble}.
\end{lemma}
\begin{proof}
  This result follows from a direct inspection of the graph support of the spin-networks $\mathcal{HH}_{n+1}$ and $\mathcal O_n$. In fact, $\mathcal{HH}_{n+1}$ is obtained from $\mathcal{H}_{n+1}$ by discarding the upper hexagonal cells. But then, attaching $\mathcal O_n$ re-constructs the missing hexagonal cells.
\end{proof}

\begin{theorem}\label{thm:inductive_step}
Let $\mathcal H_{n+1}$ denote a honeycomb of size $n+1$, and let $\bar a, \bar b, \bar c, \bar d, \bar e$ denote compatible spin colors according to the scheme described above. Then 
\begin{equation}
\begin{split}
\mathcal H_{n+1}(\bar a, \bar b, \bar c, \bar d, \bar e) =& \sum_{\bar i} \hat \Psi(\bar a, \bar b, \bar c, \bar d, \bar e\ |\ \bar i) \iota \hat \Psi(\bar a, \bar b, \bar c, \bar d, \bar e\ | \bar i) \\ &\times\mathcal H_{n}(\bar a, \bar b, \bar c, \bar d, \bar e),\\
\end{split}
\end{equation}
where $\mathcal H_{n}$ inherits the spin colors of $\mathcal H_{n+1}$.
\end{theorem}
\begin{proof}
  We apply the lemmas previously proved to obtain the result. To simplify notation we omit writing the labels of the spin-networks, but we will assume throughout to follow the conventions outlined above. Observe that using Lemma~\ref{lem:composition} we can write $\mathcal H_{n+1} = \mathcal H_n \circ_{\bar v} \mathcal O_n$. Then, following the convention for the spin colors established in the paragraph preceding Lemma~\ref{lem:half-bubble}, the edges connecting $\mathcal O_n$ to $\mathcal H_n$ are labeled by $e_k$, with $k = 0, \ldots , 2n-2$. So, we apply Lemma~\ref{lem:half-bubble} to these edges to obtain $\mathcal H_{n+1} = \sum_{\bar i}\Psi_{\bar i}\iota\Psi_{\bar i} \mathcal {HH}_n \circ_{\bar v} \mathcal {BO}_n$, where spin colors are intended as in the lemma. From Lemma~\ref{lem:bubble_rewriting} we have $\mathcal {BO}_n = \hat \Psi_{\bar i}\iota\hat\Psi_{\bar i} \mathcal O_{n-1}$.  Therefore, we have found that $\mathcal H_{n+1} = \sum_{\bar i} \hat\Psi_{\bar i} \iota\hat\Psi_{\bar i} \mathcal {HH}_n \circ_{\bar v} \mathcal O_{n-1}$.
  Lastly, we apply Lemma~\ref{lemma:half_octopus} to rewrite $\mathcal {HH}_n\circ_{\bar v} \mathcal O_{n-1} = \mathcal H_n$. This completes the proof.
\end{proof}

\begin{corollary}\label{cor:linear_sums}
The number of summation operations needed to evaluate $\mathcal H_n$ grows quadratically with $n$. More specifically, if $a_n$ denotes the number of summations at $n$, we have $a_n = a_{n-1} + 2n - 5$.
\end{corollary}
\begin{proof}
This is an immediate consequence of Theorem~\ref{thm:inductive_step} using induction.  In fact, at each step, i.e. for a fixed $n$, we have a sum on $2n-5$ indices. To see this, observe that from the proof  of Theorem~\ref{thm:inductive_step} we have to apply recoupling $2n-1$ times twice. The second round of recouplings does not introduce new labels in the summations, since in Lemma~\ref{lem:bubble_rewriting} there is no sum. In order to apply Lemma~\ref{lemma:half_octopus}, we need to apply Lemma~7 from \cite{KL} on the top of the spin-network, where three of the indices upon which we sum are present. This allows us to reduce the sum to one single index, and factor a summation of quantum dimensions coming from $i_0$ in the final result. Moreover, we notice that the base of $\mathcal O_n$ has a merging of $4$ labels and therefore a two more sums are suppressed. This gives the total number of $2n-5$ summation indices. Now, we have reduced our evaluation to $\mathcal H_{n-1}$, which inductively carries a summation over $a_{n-1}$ indices by induction. This completes the proof. 
\end{proof}

\begin{corollary}\label{cor:honeycomb_formula}
The evaluation of $\mathcal H_n(\bar a, \bar b, \bar c, \bar d, \bar e)$, with $n\geq 2$, is given by the formula 
\begin{eqnarray*}
\mathcal H_n(\bar a, \bar b, \bar c, \bar d, \bar e) = \sum_{k = 2}^n\hat\Psi_{\bar i_k}\iota\hat\Psi_{\bar i_k} \theta(c^2_1,e^2_0,b^2_0),
\end{eqnarray*}
where the coefficients $\hat\Psi_{\bar i_k}$ were given above and the index $k$ refers to the iteration of the application of Theorem~\ref{thm:inductive_step}.
\end{corollary}
\begin{proof}
We proceed by induction over $n$. For $n = 2$ we evaluate the spin-network $\mathcal H_2$ directly. Apply Lemma~\ref{lem:bubble-move} to the top of the spin-network, where we indicate the top spin color by $t$. The other that take part in the application of the lemma are, following the previously described conventions, $c^2_1, b^2_0, c^2_, b^2_2$ and $e^2_0$ which take the places of $a, d, c, b$ and $f$, respectively, in the lemma. Then we obtain
$$
\mathcal H_2 = \Delta_{e^2_0} \begin{Bmatrix}
        c^2_1 & b^2_1 & e^2_0\\
        c^2_0 & b^2_0 & t 
        \end{Bmatrix}
        \begin{tikzpicture}[scale=0.7,baseline={([yshift=-0.1cm]current bounding box.center)},vertex/.style={anchor=base,
        circle,fill=black!25,minimum size=18pt,inner sep=1pt}]
        \draw (0,0) -- (0,2);
        \draw (0,2) ..controls (1.5,0).. (1,-1);
        \draw (0,2) ..controls (-1.5,0).. (-1,-1);
        \draw (0,0) -- (-1,-1);
        \draw (0,0) -- (1,-1);
        \draw (-1,-1) ..controls (0,-2).. (1,-1);
        \node (a) at (-1.75,0) {$c^2_1$};
        \node (a) at (1.75,0) {$b^2_2$};
        \node (a) at (-0.5,-0.25) {$a^2_0$};
        \node (a) at (0.5,-0.25) {$d^2_0$};
        \node (a) at (0.25,0.75) {$e^2_0$};
        \node (a) at (0,-2) {$g$};
        \end{tikzpicture}.
$$
A second application of Lemma~\ref{lem:bubble-move}, this time with $g$ playing the role of $e$ in the diagram of the lemma, we find that 
$$
\mathcal H_2 = \Delta_{e^2_0} \begin{Bmatrix}
        c^2_1 & b^2_1 & e^2_0\\
        c^2_0 & b^2_0 & t 
        \end{Bmatrix}
        \begin{Bmatrix}
        b^2_1 & c^2_1 & e^2_0\\
        d^2_0 & a^2_0 & g 
        \end{Bmatrix}
        \theta(c^2_1,e^2_0,b^2_0),
$$
which concludes the proof of the base of induction. To derive the general formula, now we apply Theorem~\ref{thm:inductive_step} to reduce the case of dimension $n+1$ to $n$, where $n=2$ reduces to a $\theta$-net as just shown above. With the stratified labelings introduced above, to pass from $\mathcal H_{n+1}$ to $\mathcal H_n$ we need to sum over all the $\bar i$. Once we have reduced the size of $\mathcal H_n$ by one degree, we apply again Theorem~\ref{thm:inductive_step} until we reach the $n=2$ case. Each time, we relabel all the spin-colors by $\bar a_k, \bar b_k, \bar c_k, \bar d_k, \bar e_k$ to reapply all the formulas. This completes the proof. 
\end{proof}

\begin{lemma}\label{lem:inner_product}
Let $\mathcal H_1$ and $\mathcal H_2$ be honeycomb spin-networks, and let $\mathcal P$ be as above. Then, the physical inner product is given by 
$$
\langle \mathcal H_2 | \mathcal H_1\rangle_{\rm Phys} = \langle \mathcal H_1\rangle \langle \mathcal H_2\rangle,
$$
where $\langle \mathcal H_j\rangle$ indicates the evaluation computed in Section~\ref{sec:Honey}.
\end{lemma}
\begin{proof}
We apply the gauge fixing identity and the summation identity (see \cite{NP,GOP,Oeckl}) repeatedly to eliminate all the Haar integration boxes in the bulk. The case on $2\times 2$ honeycomb $\mathcal H_2$ spin-networks is shown in Figure~\ref{fig:eliminate_bulk}. In this case one proceeds as follows. First, making use of the integration boxes in the perimeter it is possible to eliminate the diagonal Haar integration boxes, as it is explicitly done for the top-left diagonal box via the blue dashed line in Figure~\ref{fig:eliminate_bulk}. Then, we can draw a circle that intersects the spin-networks only horizontally through the central integration boxex, shown in Figure~\ref{fig:eliminate_bulk} as a dotted red line. This allows us to eliminate the central box. Now, only the perimeter boxes are left, and they have a summation over the projector lines where no other integration box appear. We can therefore apply the summation identity to eliminate them, and therefore decouple the spin-netoworks completing the $2\times 2$ case. The figure shows a transition between hexagonal spin-networks where the initial and final states are superposed. The effect of the projector is that of adding lines for colors $k$ compatible with the spin colors of the states, and Haar integration (black) boxes on the edges. 

We observe that from the case $n\times n$ with $n=3$ one complication is easily seen to arise. In fact, it is not possible to directly eliminate all the boxes in the bulk by only utilizing the gauge fixing identity. It is inf fact possible to eliminate only one box per horizontal row (which in the $2\times 2$ case happens to eliminate the only horizontal box). However, since the diagonal rows are eliminated via gauge fixing, the horizontal rows can be cleared by an application of the summation identity. The perimeter is likewise cleared of any Haar integration boxes.
 
Although the previous argument provides a relatively detailed argument, we present here the general by induction using the decomposition $\mathcal H_{n+1} = \mathcal H_n\circ \mathcal O_n$ from Lemma~\ref{lem:composition} for the sake of completeness. In addition, this approach is practically useful for the implementation of the algorithm, which takes advantage of the hierarchic structure of the honeycomb spin-networks.

In practice, we use the inductive step to remove the integration boxes from the bulk of $\mathcal H_n$, and then use the gauge fixing identity between the integration boxes of $\mathcal O_n$ and those boxes in the perimeter of $\mathcal H_n$ that are in the bulk of $\mathcal H_{n+1}$. Observe that when decomposing $\mathcal H_{n+1}$, the top edge of $\mathcal H_n$ is split in two by a vertex connected with $\mathcal O_n$, so the induction is not immediately applicable. However, this is not a problem as the two integration boxes that arise on the two sides of the top vertex abut in an external cell, so that any line that is drawn through them can go out of $\mathcal O_n$ without intersecting the spin-network at a point other than a Haar integration box. So, the inductive procedure can be applied with the slight modification of using the gauge fixing identity to delete the diagonal integration boxes with two top integration boxes rather than a single one. The reader can convince themselves directly of the veracity of this assertion by drawing the connecting part of $\mathcal H_n$ and $\mathcal O_n$. Now, we observe that the only horizontal integration box that is left (on the central vertical leg of $\mathcal O_n$) can be removed by another application of the gauge fixing identity. The two aforementioned diagonal boxes on top of $\mathcal H_n$ cannot be eliminated directly both, but just one of them via gauge fixing. However, the remaining one, which is now the only non perimeter box that is left is eliminated via the summation identity, since no other box appear in the top cell of $\mathcal H_n$. The remaining perimeter boxes are eliminated once again via summation identity completing the inductive step. 
  
  Now, using the definition of the Perez-Noui projector $\mathcal P$ by means of the Ashtekar-Lewandowski measure we evaluate the spin networks in the identity element of $SU(2)$ in the classical case, while we apply them on an element $H^{-1}$ in the quantum case, where $H^{-1}$ reproduces the quantum recoupling theory \cite{CosmK}. This gives us the evaluation of $\mathcal H_j$, $j=1,2$, from Section~\ref{sec:Honey} as stated.  
\end{proof}

\bibliographystyle{apsrev4-2}
\bibliography{biblio}

\end{document}